\documentclass[11pt]{article}

\usepackage[utf8]{inputenc} 
\usepackage[T1]{fontenc}
\usepackage{lmodern}
\usepackage{dsfont}
\usepackage{geometry}
\usepackage{physics}

\usepackage{bbm}
\usepackage{algorithmic} 
\usepackage{microtype}  
\usepackage{array}
\usepackage{multirow}
\usepackage{amsmath} 
\usepackage{amssymb}
\usepackage{amsthm} 
\usepackage{thmtools}

\usepackage[procnumbered,ruled,vlined,linesnumbered]{algorithm2e}
 
\SetCommentSty{mycommfont} %

\usepackage{xcolor}  
\usepackage{xspace}
\usepackage{enumitem}    
\usepackage{tikz}   

\usepackage{comment} 

\usepackage{graphicx}
\graphicspath{{Figures/}}
\usepackage{caption}

\usepackage{subcaption}

\usepackage{xcolor}
\usepackage{nameref}
\definecolor{ForestGreen}{rgb}{0.1333,0.5451,0.1333}
\definecolor{DarkRed}{rgb}{0.65,0,0}
\definecolor{Red}{rgb}{1,0,0}
\usepackage[linktocpage=true,
pagebackref=true,colorlinks,
linkcolor=DarkRed,citecolor=ForestGreen,
bookmarks,bookmarksopen,bookmarksnumbered]
{hyperref}
\usepackage{cleveref}

\declaretheorem[numberwithin=section]{theorem}

\declaretheorem[numberlike=theorem,name=Lemma]{lem}
\declaretheorem[numberlike=theorem]{fact}

\declaretheorem[numberlike=theorem,name=Proposition]{prop}

\declaretheorem[numberlike=theorem]{claim}

\crefname{algorithm}{Algorithm}{Algorithms}
\Crefname{algorithm}{Algorithm}{Algorithms}

\theoremstyle{definition}
\declaretheorem[numberlike=theorem]{definition}

\newcommand{\ot}{\tilde{O}}
\newcommand{\Otil}{\tilde{O}}

\newcommand{\poly}{\operatorname{poly}} 
\newcommand{\polylog}{\operatorname{polylog}}

\SetKwFor{RepTimes}{repeat for}{times}{end}

\newcommand{\lo}{{\rm low}}
\newcommand{\eat}[1]{}

\newcommand{\sketch}{\operatorname{sk}}

\newcommand{\AMS}{\ell_2}

\newcommand{\textout}{\operatorname{out}}

\newcommand{\textflow}{\operatorname{flow}}

\newcommand{\textin}{\operatorname{in}}
   
\newcommand{\textdeg}{\operatorname{deg}}

\renewcommand{\Pr}{\mathbb{P}}    

\global\long\def\ltil{\widetilde{\ell}}
\global\long\def\low{\mathrm{low}}
\global\long\def\defeq{\triangleq}
\global\long\def\rlx{\mathrm{relax}}
\global\long\def\R{\mathbb{R}}
\global\long\def\ltwo{\ell_{2}}
\global\long\def\one{\mathds{1}}
\global\long\def\bad{\mathrm{bad}}
\global\long\def\Ntil{\widetilde{N}}

\global\long\def\Ztil{\widetilde{Z}}

\global\long\def\Etil{\widetilde{E}}
\global\long\def\Ftil{\widetilde{F}}
\global\long\def\outneigh{\textsc{OutNeighbor}}
\global\long\def\queue{\textsc{Queue}}
\global\long\def\countList{\textsc{CountList}}

\title{Vertex Connectivity in Poly-logarithmic Max-flows}

\author{
        Jason Li\thanks{\texttt{jmli@cs.cmu.edu}. Carnegie Mellon University, USA}
        \and
        Danupon Nanongkai\thanks{\texttt{danupon@gmail.com}. University of Copenhagen, Denmark and KTH Royal Institute of Technology, Sweden}
        \and 
        Debmalya Panigrahi\thanks{\texttt{debmalya@cs.duke.edu}. Duke University, USA}
        \and
        Thatchaphol Saranurak\thanks{\texttt{thsa@umich.edu}. University of Michigan, USA}
        \and
        Sorrachai Yingchareonthawornchai\thanks{\texttt{sorrachai.yingchareonthawornchai@aalto.fi}. Aalto University, Finland}
}

\date{}   
\begin{document}  
  
\maketitle
        \pagenumbering{gobble}
        \begin{abstract}
The vertex connectivity of an $m$-edge $n$-vertex undirected graph is the smallest number of vertices whose removal disconnects the graph, or leaves only a singleton vertex. In this paper, we give a reduction from the vertex connectivity problem to a set of maxflow instances. Using this reduction, we can solve vertex connectivity in $\ot(m^{\alpha})$ time for any $\alpha \ge 1$, if there is a $m^{\alpha}$-time maxflow algorithm. Using the current best maxflow algorithm that runs in $m^{4/3+o(1)}$ time (Kathuria, Liu and Sidford, FOCS 2020), this yields a $m^{4/3+o(1)}$-time vertex connectivity algorithm. This is the first improvement in the running time of the vertex connectivity problem in over 20 years, the previous best being an $\ot(mn)$-time algorithm due to Henzinger, Rao, and Gabow (FOCS 1996). Indeed, no algorithm with an $o(mn)$ running time was known before our work, {\em even if we assume an $\ot(m)$-time maxflow algorithm}.

Our new technique is robust enough to also improve the best $\Otil(mn)$-time bound for  \emph{directed} vertex connectivity to $mn^{1-1/12+o(1)}$ time
\end{abstract}
 
        \newpage        
    	\setcounter{tocdepth}{2}  
    	\tableofcontents        
        \newpage
        
        \pagenumbering{arabic}       
        \section{Introduction}

The {\em vertex connectivity} of an undirected graph is the size of the minimum vertex cut, defined as the minimum number of vertices whose removal disconnects the graph (or becomes a singleton vertex). %
Finding the vertex connectivity of a graph is a fundamental problem in combinatorial optimization, and has been extensively studied since the 1960s. It is well-known that the related problem of an $s$-$t$ vertex mincut, defined as the minimum vertex cut that disconnects a specific pair of vertices $s$ and $t$, can be solved using an $s$-$t$ maxflow algorithm. This immediately suggests a natural starting point for the vertex connectivity problem, namely use $O(n^2)$ maxflow calls to obtain the $s$-$t$ vertex mincuts for all pairs of vertices, and return the smallest among them. It is against this baseline that we discuss the history of the vertex connectivity problem below. Following the literature, we use $m$, $n$, and $k$ to respectively denote the number of edges, vertices, and the size of the vertex mincut in the input graph.

In the 60s and 70s, several algorithms~\cite{Kleitman1969methods,Podderyugin1973algorithm,EvenT75}, showed that for constant values of $k$, only $O(n)$ maxflow calls suffice, thereby improving the running time for this special case. The first unconditional improvement over the baseline algorithm was obtained by Becker {\em et al.}~\cite{BeckerDDHKKMNRW82}, when they used $O(n \log n)$ maxflow calls to solve the vertex connectivity problem. The following simple observation underpinned their new algorithm: if one were able to identify a vertex $s$ that is not in the vertex mincut, then enumerating over the remaining $n-1$ vertices as $t$ in the $s$-$t$ maxflow calls is sufficient. They showed that they could obtain such a vertex $s$ whp\footnote{with high probability} by a random sampling of vertices.

The next round of improvement was due to Linial, Lov\'asz, and Wigderson (LLW)~\cite{LinialLW88} who used an entirely different set of techniques based on matrix multiplication to achieve a running time bound of $O((n^{\omega} + nk^{\omega})\log n)$, which is $O(n^{1+\omega}\log n)$ in the worst case of $k = \Theta(n)$; here, $\omega\approx 2.37$ is the matrix multiplication exponent.
To compare this with the maxflow based algorithms, we note that the maxflow instances generated by the vertex connectivity problem are on unit vertex-capacity graphs, for which an $O(m\sqrt{n})$ algorithm has been known since the celebrated work of Dinic using blocking flows in the 70s~\cite{Dinic70}. Therefore, LLW effectively improved the running time of vertex connectivity from $O(n^{7/2})$ in the worst case to $O(n^{1+\omega})$.

A decade after LLW's work, Henzinger, Rao, and Gabow (HRG)~\cite{HenzingerRG00} improved the running time further to $O(mn\log n)$ by reverting to combinatorial flow-based techniques. They built on the idea of computing $O(n)$ maxflows suggested by Becker {\em et al.}~\cite{BeckerDDHKKMNRW82}, but with a careful use of preflow push techniques~\cite{GoldbergT88} in these maxflow subroutines, they could amortize the running time of these maxflow calls (similar to, but a more refined version of, what Hao and Orlin had done for the edge connectivity problem a few years earlier~\cite{HaoO94}). The HRG algorithm remained the fastest unconditional vertex connectivity algorithm before our work. 

We also consider the vertex connectivity problem on directed graphs. Here, the goal is to find a smallest set of vertices whose removal ensures that the remaining graph is not strongly connected. The HRG bound of $O(mn\log n)$~\cite{HenzingerRW17} generalizes to digraphs, and sets the current record for this problem as well.

In concluding our tour of vertex connectivity algorithms, we note that there has also been a large volume of work focusing on faster algorithms for the special case of {\em small $k$}. Nearly-linear time algorithms are known only when $k \le 2$ \cite{Tarjan72,HopcroftT73,KanevskyR91,NagamochiI92,CheriyanT91,Georgiadis10} until recently when \cite{NanongkaiSY19,ForsterNYSY20} give an $\ot(mk^2)$-time algorithm\footnote{$\ot(f(n)) = O(\polylog(n)f(n)).$} for both undirected and directed graphs, which is  nearly-linear for $k = \polylog (n)$.
Similarly, the question of {\em approximating} the vertex connectivity of a graph efficiently has received some attention, and a $(1+\epsilon)$-approximation is known in $\ot(\min\{mk/\epsilon,n^{\omega}/\epsilon^2\})$ time~\cite{NanongkaiSY19,ForsterNYSY20} while a worse approximation factor of $O(\log n)$ can be achieved in near-linear time~\cite{Censor-HillelGK14}. These two lines of work are not directly related to our paper.

\subsection{Our Results}

In this paper, we give the following result:

\begin{theorem}[Main]
\label{thm:main}
Given an undirected graph on $m$ edges, there is a randomized, Monte Carlo vertex connectivity algorithm that makes $s$-$t$ maxflow calls on unit capacity graphs that cumulatively contain $\ot(m)$ vertices and $\ot(m)$ edges, and runs in $\ot(m)$ time outside these maxflow calls.
\end{theorem}

\noindent
In other words, if maxflow can be solved in $m^\alpha$ time on unit capacity graphs, for any $\alpha\geq 1$, then we can solve the vertex connectivity problem in $\tilde O(m^\alpha)$ time. In particular, using the current fastest maxflow algorithm on unit capacity graphs (Kathuria, Liu and Sidford~\cite{KathuriaLS20}), we get a vertex connectivity algorithm for undirected graphs that runs in $m^{4/3+o(1)}$ time, which strictly improves on the previous best time complexity of $\ot(mn)$ achieved by the HRG algorithm. Even more ambitiously, if maxflow is eventually solved in $\ot(m)$ time, as is often conjectured, then our theorem will automatically yield an $\ot(m)$ algorithm for the vertex connectivity problem, which would resolve the long standing open question by Aho, Hopcroft and Ullman \cite{AhoHU74} since 1974 up to polylogarithmic factors. In contrast, even with an $\ot(m)$-time maxflow algorithm, no previous vertex connectivity algorithm achieves an $o(mn)$ running time bound.

We remark that the reduction in the theorem generates instances of the $s$-$t$ vertex connectivity problem, i.e., a maximum set of vertex-disjoint paths between $s$ and $t$ in an undirected graph, which are solved by a maxflow call via a standard reduction. Also, we note that our algorithm is randomized (Monte Carlo) even if the maxflow subroutines are not. It is an interesting open question to match the running time bounds of this theorem using a deterministic algorithm, or even a Las Vegas one.

We also generalize our new technique to work \emph{directed}
graphs and obtain a significant improvement upon the fastest $\ot(mn)$-time algorithm by HRG for the directed vertex connectivity problem
\begin{restatable}{theorem}{directedConn}
	\label{thm:main directed intro}
	Given a directed graph with $m$ edges and $n$ vertices, there are
	randomized Monte Carlo vertex connectivity algorithms with 
	\begin{itemize}
		\item $mn^{1-1/12+o(1)}$ time, or
		\item $\ot(n^{2})$ time assuming that max flow can be solved in near-linear
		time.
	\end{itemize}
\end{restatable}
As the result on directed graphs is obtained by using our new technique in a less efficient way and does not give additional insight, we discuss it in the Appendix.

        \subsection{Technical Overview}
\label{sec:overview}

Our main technical contribution is a new technique that we call {\em sublinear-time kernelization} for vertex connectivity. Namely, we show that under certain technical conditions, we can find a subgraph whose size is sublinear in $n$ and preserves the vertex connectivity of the original graph. We use sketching techniques to construct such a subgraph in sublinear time. 
To the best of our knowledge, all previous techniques require $\Omega(n^{3})$
time even in the extremely unbalanced case when the vertex mincut have size $\Omega(n)$, and the smaller side of the mincut contains $O(1)$ vertices. In contrast, sublinear-time kernelization allows us to reduce the problem in this
case to maxflow calls of total size $\ot(m)$ in $\ot(m)$ time.
Below, we elaborate on this new technique and discuss how it fits into the entire vertex connectivity algorithm.

Suppose the vertex mincut of the input graph $G$ is denoted by $(L, S, R)$, where $|L| \le |R|$ are the two sides of the cut, and $|S| = k$ is the set of vertices whose removal disconnects $L$ from $R$. For intuitive purposes, let us assume that we know the values of $|L|$ and $|R|$ and $|R|=\Omega(n)$.
This allows us to obtain a vertex in $R$ using just $O(\log n)$ samples.
From now, we assume that we know a vertex $r\in R$.
If we were also able to find a vertex $x\in L$, then we can simply compute an $x$-$r$ maxflow to obtain a vertex mincut. But, in general, $|L|$ can be small, and obtaining a vertex in $L$ whp requires $\ot(n/|L|)$ samples. Recall that we promised that the total number of edges in all the maxflow instances that we generate will be $\ot(m)$. One way of ensuring this would be to run each of the $\ot(n/|L|)$ maxflow calls on a graph containing only $\ot(k|L|)$ edges; then, the total number of edges in the max flow instances is $\ot((n/|L|)\cdot (k |L|)) = \ot(nk) = \ot(m)$ since the degree of every vertex is at least $k$. %
At first glance, this might sound impossible because the number of edges incident to $N(x)$ is already $\Omega(k^2)$.  Nevertheless, our main technical contribution is in showing that in certain cases we can construct a graph $H$ with just $\ot(k|L|)$ edges that gives us information about the vertex connectivity of $G$. We call such graph a {\em kernel}.
In achieving this property, we need additional conditions on $L$ and $S$, specifically on their relative sizes and the degrees of vertices in $S$.
If these conditions do not hold, we give a different algorithm that uses a recent tool called the isolating cut lemma used in the edge connectivity problem~\cite{LiP20deterministic} (we adapt the tool to vertex connectivity). 
More specifically, we consider three cases depending on the sizes of $L$ and $S_\lo$, where 
$$S_\lo=\{v\in S \mid \deg(v)\leq 8k\}.$$ 
It might not be intuitive now why we need $S_\lo$. Distinguishing cases using $S_\lo$ is a crucial idea that makes everything fits together. Its role will be more clear in the discussion below. The use of our kernelization is in the last case (Case 3). We now discuss all the cases.

\paragraph{Case 1: Large $L$ (details in \Cref{sec:using isolating}).}
We first consider the easier case when $L$ is not too small compared to $S$, i.e. $|L| > k/\polylog(n)$.
Consider the vertex set $T$ where each vertex is included in $T$ with probability $1/k$. Then, with probability at least $1/\polylog(n)$, $T$ contains {\em exactly one} vertex from $L$ (call it $x$), no vertex from $S$, and the remaining vertices are from $R$. Assume that it is the case by repeating $\polylog(n)$ times. Observe that a vertex mincut separating $x$ from $T\setminus \{x\}$, denoted by $(x,T\setminus\{x\})$-vertex mincut, is a (global) vertex mincut of $G$. 

The {\em isolating cut lemma} was recently introduced by Li and Panigrahi~\cite{LiP20deterministic} for solving the edge connectivity problem deterministically. It says that in an undirected graph, given a set of terminal vertices $T$, we can make maxflow calls to graphs of total size $O(m \log |T|)$ and return for each terminal $t\in T$, the smallest \emph{edge} cut separating $t$ from $T\setminus \{t\}$. In particular, it returns us a  $(x,T\setminus\{x\})$-edge mincut.
If this lemma worked for vertex cuts, it would return us a $(x,T\setminus\{x\})$-vertex mincut and we would be done. 
It turns out that the isolating cuts lemma can be adapted to work for vertex connectivity, due to the submodularity property of vertex cuts. 

\paragraph{Case 2: Small $L$, small $S_\lo$ (details in \Cref{sec:using isolating}).}
From now on, we assume that $|L| < k/\polylog(n)$. 
Note that for every vertex $x \in L$, all neighbors of $x$ are inside $L \cup S$ and so $\deg(x)\le |L|+k < 2k$. Let $V_\low$ be all vertices whose degrees are less than $8k$. We know that $L\subseteq V_\low$ and $S_\lo = S\cap V_\lo$ by definition. It is also easy to show that $|R\cap V_\lo|\geq |L|$ (see \Cref{claim:helper low}).
So, if $|S_\lo| < |L| \cdot \polylog(n)$, then, by sampling from $V_\lo$ instead of $V$ with probability $1/(|L|\polylog(n))$, we can obtain a random sample that includes exactly one vertex from $L$, some vertices from $R$, and none from $S$, as in the previous case. In this case, we again can apply the isolating cuts lemma.

\paragraph{Case 3: Small $L$, large $S_\lo$ (details in \Cref{sec:kernel}).}
The above brings us to the crux of our algorithm, where the isolating cuts lemma is no longer sufficient. Namely, $L$ is much smaller than the cut $S$ and $S$ contains many vertices with low degree, i.e. 
\begin{align}
|L| < k/\polylog(n)\mbox{ and } |S_\lo| > |L|\cdot \polylog(n).\label{eq:overview:case3}
\end{align}
Let us first sample $\ot(n/|L|)$ vertices; at least one of these vertices is in $L$ whp. 
Now, for each vertex $x$ in the sample, we will invoke a maxflow instance on $\ot(k |L|)$ edges that returns the vertex mincut if $x\in L$.
This suffices because $\ot((n/|L|)\cdot (k |L|)) = \ot(nk)$ can be bounded by $\ot(m)$, noting that the degree of every vertex is at least $k$. 
Thus, we can reduce our problem to the following goal: 
\begin{quote}
{\em Given a vertex $x\in L$, describe a procedure to create a maxflow instance on $\ot(k |L|)$ edges that returns a vertex mincut.}
\end{quote}
In other words, assuming that we have a vertex $x\in L$, we want to construct a small graph $H$ and two vertices $s$ and $t$ in $H$ such that the $(s,t)$-maxflow in $H$ tells us about the vertex mincut in the original input graph. The graph $H$ corresponds to the concept of {\em kernel} in parameterized algorithms. 
A challenge is that it is not clear if a small kernel exists for vertex connectivity; it is not even clear if it is possible to reduce the number of edges at all. The entire description below aims to show that it is possible to reduce the number of edges to $\ot(k|L|)$. We ignore the time complexity for this process for a moment.

The key step is to define the following set $T_x$. First, let $T$ be a set such that every vertex is in $T$ with probability $1/|L|$. 
Then, $T_x$ is defined from $T$ by {\em excluding} $x$ and its neighbors, i.e. $T_x=T\setminus N_G[x]$, where $N_G[x]=N_G(x)\cup \{x\}$ and $N_G(v)$ denotes the set of neighbors of $v$. (We drop $G$ when the context is clear).
We exploit a few properties of $T_x$.
First, we claim that $T_x\subseteq R$  with $\Omega(1)$ probability. 
To see this, note that $N[x]\subseteq L\cup S$ for any  $x\in L$. Since $|N[x]|>k$ but $|L\cup S|\le|L|+k$ , it must be the case that 
\begin{align}
|(L\cup S)\setminus N[x]|<|L|. \label{eq:overview:small S_x}
\end{align}
Now, $T_{x}\subseteq R$ iff none of vertices from $(L\cup S)\setminus N[x]$ is sampled to $T$. As $|(L\cup S)\setminus N[x]|<|L|$ and the sampling probability is $1/|L|$, so $T_{x}\subseteq R$ with $\Omega(1)$ probability.

From now we assume that $T_x\subseteq R$. Consider contracting vertices in $T_x$ into a single node $t_x$. %
Since $T_x\subseteq R$, an $(x, t_x)$-maxflow call would return a vertex mincut of the original graph. However, the contracted graph might still contain too many edges. 
To resolve this issue, we make the following important observations: 
\begin{enumerate}[noitemsep]
    \item any vertex $v$ neighboring to both $x$ and $t_x$ must be in $S$, and 
    \item there exists a collection of $k$ vertex disjoint paths between $x$ and $t_x$ where each path contains exactly one neighbor of $x$ and exactly one neighbor of $t_x$. %
\end{enumerate}
The observations above simply follow from the fact that $x$ and $t_x$ are on the different side of the vertex mincut.
The first observation allows us to remove all common neighbors of $x$ and $t_x$ and add them back to the vertex mincut later. The second observation allows us to remove all edges between neighbors of $x$ and all edges between neighbors of $t_x$ without changing the $(x, t_x)$ vertex connectivity. Further, after all these removals, neighbors of $t_x$ of degree one (i.e. they are adjacent only to $t_x$) can be removed without changing the $(x,t_x)$ vertex connectivity. 
Interestingly, these removals are already enough for us to show that there are $\tilde O(k|L|)$ vertices and edges left!

\paragraph{Small kernel.} We call the remaining graph from above a {\em kernel} and denote it by $H$. We now show that $H$ contains $\tilde O(k|L|)$ edges whp.
Note that $H$ consists of the terminals $x$ and $t_x$, disjoint sets $N_x\defeq N_H(x)$ and $N_t\defeq N_H(t_x)$, and all other vertices in a set that we call $F'$ (for ``far''). We illustrate this in \Cref{fig:overview}. 

\begin{figure}[!h] 
\centering
\includegraphics[width=0.35\textwidth ]{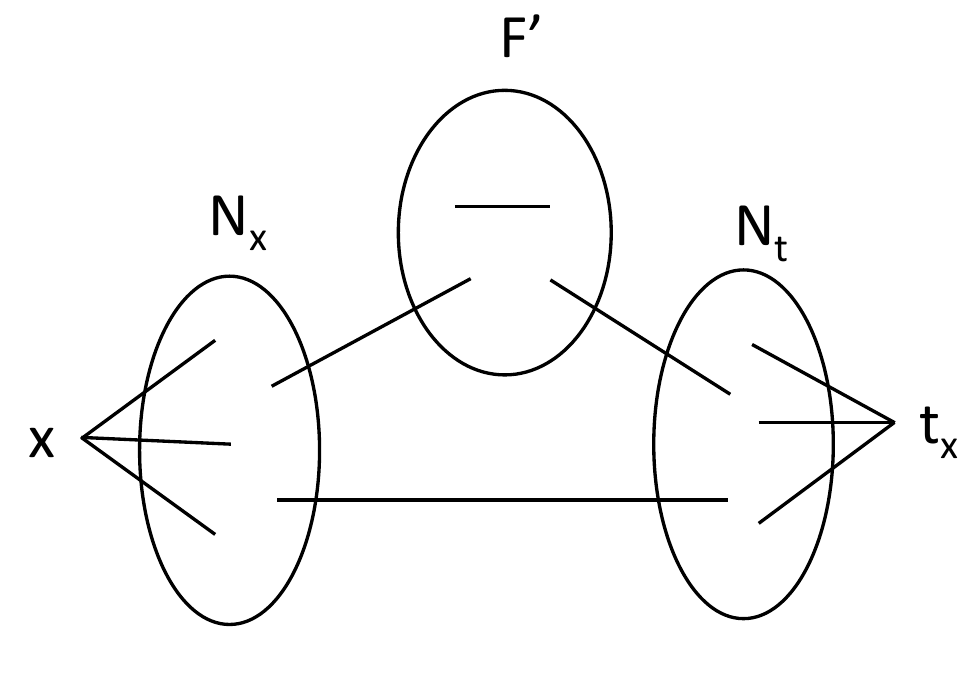} 
\caption{Kernel $H$ of graph $G=(V,E)$. } 
\label{fig:overview}
\end{figure}
Recall that we have already discarded all internal edges in $N_x$ and $N_t$; hence, we have three types of edges in $H$:
\begin{itemize}[noitemsep]
    \item[(E1)] edges in $N_x\times (F'\cup N_t)$, i.e. edges with  one endpoint in $N_x$ and the other in $F'$ or $N_t$, 
    \item[(E2)] edges in $F'\times (F'\cup N_t)$, i.e. edges with one endpoint in $F'$ and the other in $F'$ or $N_t$, and
    \item[(E3)] edges incident to terminals $x$ and $t_x$. 
\end{itemize}
We count the number of edges in (E1) and (E2) by charging them to its endpoint in $N_x$ and $F'$ respectively. We will show that there are $\tilde O(k|L|)$ such edges in total. It then follows that there are  $\tilde O(k|L|)$ edges in (E3), since there are at most $k+|L|$ edges incident to $x$ and each vertex in $N_t$ must be incident to some edge in (E1) or (E2) (otherwise, we would have already deleted such vertex). 
The claimed $\tilde O(k|L|)$ bound on the number of edges in (E1) or (E2) follows immediately once we show that whp
\begin{itemize}[noitemsep]
    \item[(a)] every vertex in $N_x\cup F'$ is charged by $\ot(L)$ edges, and 
    \item[(b)] there are $O(k)$ vertices in $N_x\cup F'$. 
\end{itemize}
To prove (a), consider any vertex $v\neq x$ in $G$ with $\deg_{G\setminus N[x]}(v)>|L| \cdot \polylog(n)$, i.e. $v$ has many neighbors outside $N_G[x]$, the neighborhood of $x$. Then, one of these neighbors must have been sampled to $T$ whp, and would be retained in $T_x$. This implies that such $v$ is in $N_t$ whp. This implies further that every vertex $v\in N_x\cup F'$ has at most $|L|\cdot \polylog(n)$ edges to vertices in $F'\cup N_t$ (since the latter vertices are all outside of $N_G[x]$). 
This establishes (a).

To prove (b), first note that $|N_x| < |L| + |S| < 2k$ since  $N_x \subseteq N_G(x)\subseteq L\cup S$; so, it is left to show that $|F'|=O(k)$. 
The key statement that we need is that
\begin{align}
\mbox{$\Omega(|S_\lo|)$ neighbors of every vertex $v\in F'$ are in $S_\lo$.}\label{eq:overview:bound F}
\end{align}
Given this, as we know that vertices in $S_\lo$ are incident to $O(k|S_\lo|)$ edges in total, we have $|F'| = O(\frac{k |S_\lo|}{|S_\lo|}) = O(k)$ as desired. 
To prove \eqref{eq:overview:bound F}, we essentially use the following facts (for precise quantities, see \Cref{fig:small-F}). 
\begin{itemize}[noitemsep]
    \item[(b1)] There are less than $k+|L|$ vertices in $N_x$, i.e. $|N_x|<k+|L|$ (we just proved this above).
    \item[(b2)] All but $|L|$ of vertices in $S_\lo$ are in $N_x$. This follows from \eqref{eq:overview:small S_x}.
    \item[(b3)] Whp, every vertex in $F'$ has at least $k-|L|\polylog n$ neighbors in $N_x$. This follows from the argument in the proof of (a).  
\end{itemize}
This means that each $v\in F'$ has at least the following number of neighbors in $S_\lo$:
$$k-|L|\polylog(n)-(|N_x|-(|S_\lo|-|L|))\geq k-|L|\polylog(n)-(k+|L|)+(|S_\lo|-|L|)=\Omega(|S_\lo|)$$
 where the last equality holds as $k$'s cancel each other and $|S_{\low}|$ dominates other terms. 
This is the crucial place where we need that $S_\low$ is large as stated in \eqref{eq:overview:case3}. Without this guarantee, we could not have bounded the size of $H$ and this explains the reason why we need to introduce Case 2 above.
 This completes the proof of (b).

\begin{figure}[!h] 
\centering
\includegraphics[width=0.5\textwidth ]{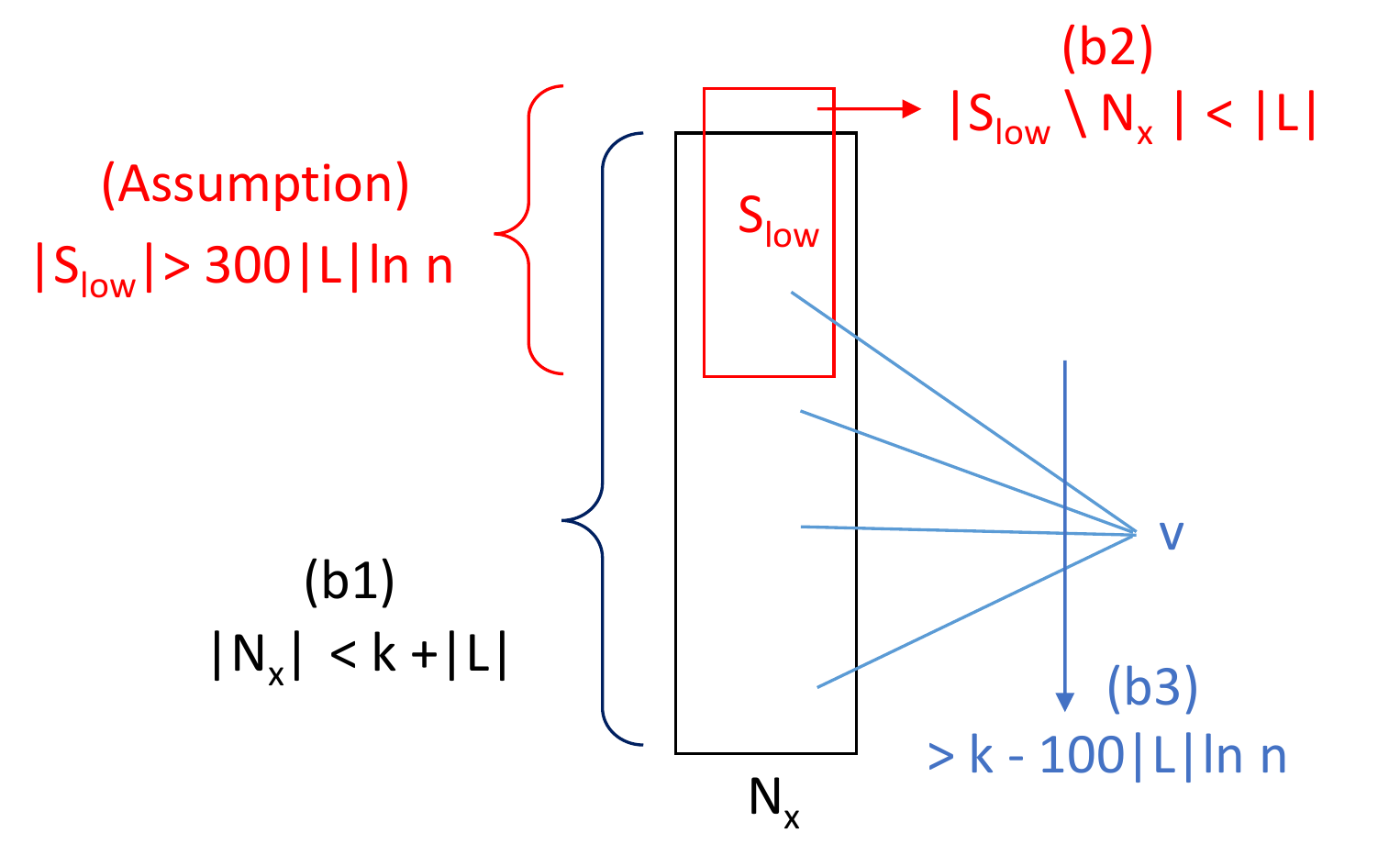} 
\caption{Facts (b1), (b2) and (b3).} 
\label{fig:small-F}
\end{figure}

\paragraph{Building kernels in sublinear time.}
So far, we only bound the size of the kernel $H$. Below, we discuss how to actually build it in sublinear time. Note that we will end up building a {\em subgraph} of $H$ instead of $H$.

Consider the following BFS-like process: Initialize the queue of the BFS with
vertices in $N_{x}$. Whenever $v$ is visited, if $v\notin N_{t}$,
we add $N(v)\setminus N[x]$ into the queue. This process will
explore the ``relevant'' subgraph of $H\setminus\{x,t_{x}\}$ because
the part that is not even reached from $N_{x}$ cannot be relevant
to $(x,t_{x})$-vertex connectivity in $H$ and so we ignore it. The kernel
graph that our algorithm actually constructs is obtained by adding $E(x,N_{x})$
and $E(N_{t},t_{x})$ into the above explored subgraph of $H$.
Our goal is to implement this process in $\ot(k|L|)$ time. There
are two main challenges. 
\begin{enumerate}[noitemsep]
    \item[(c1)] For all $O(k)$ visited vertices $v\notin N_{t}$, we must list $N(v)\setminus N[x]$ in $\ot(|L|)$ time. Note that simply listing neighbors of $v$ already takes $\deg(v)\ge k$ time which is too expensive.
    \item[(c2)] For all $\ot(k|L|)$ visited vertices $v$, we must test if $v\in N_{t}$ (i.e.~if its neighborhood in $G$ overlaps with $T_x$) in $\polylog(n)$ time.
\end{enumerate}

We address both challenges by implementing our BFS-like process based
on linear sketches from the streaming algorithm community, and so
we call our technique \emph{sketchy search}. The key technique for (c1) is \emph{sparse recovery
sketches}: An $s$-sparse recovery sketch  \emph{linearly} maps
a vector $\vec{a}\in\mathbb{Z}^{n}$ to a smaller vector $\sketch_{s}(\vec{a})\in\mathbb{Z}^{\ot(s)}$
in $\ot(\|\vec{a}\|_{0})$ time so that, if $\vec{a}$ has at most
$s$ non-zero entries, then we can recover $\vec{a}$ from $\sketch_{s}(\vec{a})$
in $\ot(s)$ time. 
For any vertex $v$, let $\one_{N(v)}$ and $\one_{N[v]}$ be the indicator
vectors of $N(v)$ and $N[v]$ respectively. We observe two things:
(1) non-zero entries in $\one_{N(v)}-\one_{N[x]}$ correspond to the symmetric
difference $N(v)\triangle N[x]$, and (2) $|N(v)\triangle N[x]|=\Theta(|N(v)\setminus N[x]|+|L|)$ (formally proved in \eqref{eq:sym}). 

This motivates the following algorithm. Set $s\gets|L|\polylog(n)$
and precompute $\sketch_{s}(\one_{N(v)})$ and $\sketch_{s}(\one_{N[v]})$
for all vertices $v$. This takes $\sum_v\ot(\deg(v))=\ot(m)$ time. Now, given
any $v$, we can compute in $\ot(s)$ time $\sketch_{s}(\one_{N(v)})-\sketch_{s}(\one_{N[x]})=\sketch_{s}(\one_{N(v)}-\one_{N[x]})$,
where the equality is because the map is linear. If $v\notin N_{t}$, then we have argued previously that $|N(v)\setminus N[x]|\le|L|\polylog(n)$
and so $\one_{N(v)}-\one_{N[x]}$ has at most $s$ non-zero entries. Thus,
from $\sketch_{s}(\one_{N(v)}-\one_{N[x]})$ we can obtain $N(v)\triangle N[x]$
which contains the desired set $N(v)\setminus N[x]$ in $\ot(s)=\ot(|L|)$
time.

To address (c2), recall that if $|N(v)\setminus N[x]|\ge|L|\polylog(n)$,
then  $v\in N_{t}$. This condition can be checked in
$O(\log n)$ time using another linear sketch (called \emph{norm estimation}) for estimating $\|\one_{N(v)}-\one_{N[x]}\|_{2}$
which is proportional to $|N(v)\setminus N[x]|+|L|$. 
However, there
can still be some $v\in N_{t}$ but $|N(v)\setminus N[x]|\le|L|\polylog(n)$.
Fortunately, there are only $O(k)$ such vertices in $N_{t}$ (using
the same argument that bounds $|F|=O(k)$ in the proof of (b)). For those vertices, we
have enough time to list $N(v)\setminus N[x]$ in $\ot(|L|)$ time using
the sparse recovery sketches and check if $t_{x}\in N(v)\setminus N[x]$,
which holds iff $v\in N_{t}$.

\paragraph{Remarks:} 
Note that all we have established is that in any one of the many invocations of the sampling processes being used, we will return a vertex mincut. For the sake of correctness, we carefully argue later that in all the remaining calls, i.e., when sampling does not give us the properties we desire, we actually return some vertex cut in the graph. This allows us to distinguish the vertex mincut from the other cuts returned, since it has the fewest vertices.

        \section{Preliminaries}

Let $G=(V,E)$ be an undirected graph. For any set $T$ of vertices, we let $N_{G}(T)=\{v\notin T\mid\exists u\in T$
and $(u,v)\in E\}$ and $N_{G}[T]=T\cup N_{G}(T)$. If $T=\{v\}$,
we also write $N(v)$ and $N[v]$. The set $E_{G}(A,B)$ denote the
edges with one endpoint in $A$ and another in $B$. If $A=\{v\}$,
we write $E_{G}(v,B)$. We usually omit the subscript when $G$ refers
to the input graph. For any graph $H$, we use $V(H)$ to denote the set
of vertices of $H$, and $E(H)$ to denote the set of edges of $H$.  %
Whenever we contract a set of vertices in a graph, we remove all
parallel edges to keep the graph simple. This is because parallel
edges does not affect vertex connectivity. 

A \emph{vertex cut} $(L,S,R)$ of a graph $G=(V,E)$ is a partition
of $V$ such that $L,R\neq\emptyset$ and $E_{G}(L,R)=\emptyset$.
We call $S$ the corresponding \emph{separator} of $(L,S,R)$. The
size of a vertex cut is the size of its separator $|S|$. A vertex
cut $(L,S,R)$ is an $(s,t)$-vertex cut if $s\in L$ and $t \in R$.  A vertex mincut is a vertex cut with minimum size. An $(s,t)$-vertex
mincut is defined analogously. If $(L,S,R)$ is an $(s,t)$-vertex
mincut, we say that $S$ is an $(s,t)$-min-separator.  
For disjoint subsets $A,B \subset V$, a vertex cut $(L,S,R)$ is an $(A,B)$-vertex cut if $A \subseteq L$ and $B \subseteq R$. $(A,B)$-separator and $(A,B)$-min-separator are defined analogously.
Throughout the
paper, we assume wlog that $$|L|\le |R|.$$

In \Cref{sec:Fast Kernel}, we will employ the following standard linear
sketching techniques. We state the known results in the form which
is convenient for us below. We prove them in the Appendix. In both
theorems below, an input vector $v$ is represented in a sparse
representation, namely a list of (index,value) of non-zero
entries. The number of non-zero entries of $v$ is denoted as $\norm{v}_0$.
\begin{theorem} 
[Norm Estimation]\label{thm:estimate}For any number $n$, there
is an algorithm that preprocesses in $\ot(n)$ time and then, given
any vector $v\in\R^{n}$, return a sketch $\sketch_{\AMS}(v)\in\R^{O(\log n)}$
in $\ot(\|v\|_{0})$ time such that $\|v\|_{2}\le\|\sketch_{\AMS}(v)\|_{2}\le1.1\|v\|_{2}$
whp. Moreover, the sketch is \emph{linear}, i.e.~$\sketch_{\ltwo}(u+v)=\sketch_{\ltwo}(v)+\sketch_{\ltwo}(u)$
for any $u,v\in\R^{n}$. 
\end{theorem}

\begin{theorem}[Sparse Recovery]\label{thm:recovery}For any numbers $n$ and
$s$, there is an algorithm that preprocesses in $\ot(s)$ time and
then, given any vector $v\in \{-1,0,1\}^{n}$, return a sketch $\sketch_{s}(v)\in\mathbb{Z}^{\ot(s)}$
in $\ot(\|v\|_{0})$ time and guarantees the following whp (as long as
the number of recovery operations is $\poly(n)$).\footnote{The algorithm works for larger range $[-\poly(n),\poly(n)]$ of integers, but the range $\{-1,0,1\}$ is sufficient for our purpose.} 
\begin{itemize}
\item If $\|v\|_{0}\le s$, then we can recover $v$ from $\sketch_{s}(v)$
in $\ot(s)$ time. (More specifically, we obtain all non-zero entries
of $v$ together with their indices). 
\item Otherwise, if $\|v\|_{0}>s$, then the algorithm returns $\bot$.
\end{itemize}
Moreover, the sketch is \emph{linear}, i.e.~$\sketch_{s}(u+v)=\sketch_{s}(v)+\sketch_{s}(u)$
for any $u,v\in \mathbb{Z}^{n}$. 
\end{theorem}

\section{Using Sublinear-time Kernelization}\label{sec:kernel}

We say that a vertex cut $(L,S,R)$ is a \emph{$k$-scratch} if $|S|<k$, $|L|\le k/(100\log n)$ and $|S_{\low}|\ge300|L|\ln n$
where $|S_{\low}|=\{v\in S\mid\deg(v)\le 8k\}$. This kind of cuts is considered in Case 3 of \Cref{sec:overview}.
In this section, we show that if a graph has $k$-scratch, then we can return some vertex cut of size less than $k$.

\begin{lem}
\label{lem:main k scratch}There is an algorithm that, given an undirected
graph $G$ with $n$ vertices and $m$ edges and a parameter $k$,
returns a vertex cut $(L,S,R)$ in $G$. If $G$ has a $k$-scratch, then $|S|<k$ w.h.p. The algorithm makes $s$-$t$ maxflow calls on unit-vertex-capacity graphs with $\ot(m)$ total number of vertices and edges
and takes $\ot(m)$ additional time.
\end{lem}

Throughout this section, we assume that minimum degree of $G$ is at least $k$, otherwise the lemma is trivial.
The rest of this section is for proving the above lemma. Assume that
a $k$-scratch exists, let $(L,S,R)$ be an arbitrary $k$-scratch.
We start with a simple observation which says that, given a vertex
$x\in L$, the remaining part of $L\cup S$ outside $N[x]$ has size at most $|L|$ which is potentially much smaller than $k$. 
\begin{prop}
\label{prop:outside Nx}For any $x\in L$, $|(L\cup S)\setminus N[x]|<|L|$. 
\end{prop}

\begin{proof}
Note that $N[x]\subseteq L\cup S$ as $x\in L$. The claim follows
because $|L\cup S|<|L|+k$ and $|N[x]|>k$ as the minimum degree is
at least $k$.
\end{proof}
We will use $\ltil$ as an estimate of $|L|$ (since $|L|$ is actually
unknown to us). Let $T$ be obtained by sampling each vertex with
probability $1/(8\ltil)$. Let $T_{x}\defeq T\setminus N[x]$ for
any $x\in V$. 
Below, we show two basic properties of $T$. 
\begin{prop}\label{prop:low deg whp}
For any $x\in V$, we have the following whp. 
\begin{equation}
\text{For every }v\notin N[T_{x}],\,|N(v)\setminus N[x]|\le40\ltil\ln n\label{eq:low out deg}
\end{equation}
\end{prop}

\begin{proof}
It suffices to prove that, for any $v\in V$, if $|N(v) \setminus N[x]|>40\ltil\ln n$,
then $v$ is incident to $T_{x}$ whp. Indeed,
$v$ is not incident to $T_{x}$ is with probability at most $(1-\frac{1}{8\ltil})^{|N(v) \setminus N[x]|} < n^{-5}$. 
\end{proof}

\begin{prop}
\label{prop:avoid SL}Suppose  $|L|/4\le\ltil\le|L|$. For each  $x\in L$,
$\emptyset\neq T_{x}\subseteq R$ with constant probability.
\end{prop}

\begin{proof}
Note that $\emptyset\neq T_{x}\subseteq R$ iff none of vertices from
$(L\cup S)\setminus N[x]$ is sampled to $T$ \emph{and} some vertex
from $R\setminus N[x]$ is sampled to $T$. Observe that $|(L\cup S)\setminus N[x]|<|L|$
by \Cref{prop:outside Nx} and $|R\setminus N[x]|=|R|\ge|L|$.

To rephrase the situation, we have two disjoint sets $A_{1}$ and
$A_{2}$ where $|A_{1}|<|L|$ and $|A_{2}|\ge|L|$ and each element
is sampled with probability $\frac{1}{8\ltil}\in[\frac{1}{8|L|},\frac{1}{2|L|}]$.
No element is $A_{1}$ is sampled with probability at least $(1-\frac{1}{2|L|})^{|L|}\ge0.5$.
Some element in $A_{2}$ is sampled with probability at least $1-(1-\frac{1}{8|L|})^{|L|}\ge1-e^{1/8}\ge0.1$.
As both events are independent, so they  happen simultaneously with probability at
least $0.05$. That is, $\emptyset\neq T_{x}\subseteq R$ with probability at least
$0.05$.
\end{proof}
For intuition, let us see why these observations above can be useful.
Suppose $\ltil\approx|L|$ and we can guess $x\in L$. Then, \Cref{prop:avoid SL}
says that $\emptyset\neq T_{x}\subseteq R$  with some chance. This
implies that any $(x,T_{x})$-vertex mincut must have size at most
$|S|<k$ and we so could return it as the answer of \Cref{lem:main k scratch}. 
However, directly computing a $(x,T_{x})$-vertex mincut in $G$ is
too expensive. One initial idea is to contract $T_{x}$ into a single
vertex $t_{x}$ (denoted the contracted graph by $G'_{x,T}$) and then compute
a $(x,t_{x})$-vertex mincut in the smaller graph $G'_{x,T}$. Now,
\Cref{eq:low out deg} precisely means that, for every vertex $v$
in $G'_{x,T}$ not incident to the sink $t_{x}$ and not $t_{x}$
itself, the neighbor set of $v$ outside $N[x]$ is at most $40\ltil\ln n$.

This fact that many vertices in $G'_{x,T}$ has ``degree outside $N[x]$'' at most $\ot(\ltil)$
is the key structural property used for constructing a small graph $G_{x,T}$
with $\ot(k\ltil)$ edges such that a $(x,t_{x})$-vertex mincut in
$G_{x,T}$ corresponds to a $(x,T_{x})$-vertex mincut in $G$. The graph $G_{x,T}$
fits into the notion of {\em kernel} in parameterized algorithms and hence we call it a \emph{kernel} graph.
The graph $G_{x,T}$ will be obtained from $G'_{x,T}$ by removing further edges and vertices. 

The following key lemma further
shows that, given a set $X$, we can build the kernel graph $G_{x,T}$ for each $x\in X$ in $\ot(k\ltil)$ time, which is sublinear time.
\begin{lem}[Sublinear-time Kernelization]
\label{lem:kernel}Let $G$ and $k$ be the input of \Cref{lem:main k scratch}.
Let $\ltil\le k/(100\log n)$. Let $X$ be a set of vertices. Let
$T$ be obtained by sampling each vertex with probability $1/(8\ltil)$ and $T_{x}\defeq T\setminus N[x]$ for any $x\in X$.
There is an algorithm that takes total $\ot(m+|X|k\ltil)$ time such
that, whp, for every $x\in X$, either 
\begin{itemize}
\item outputs a \emph{kernel} graph $G_{x,T}$ containing $x$ and $t_{x}$
as vertices where $|E(G_{x,T})|=O(k\ltil\log n)$ together with a vertex set $Z_{x,T}$ 
such that a set $Y$ is 
a $(x,t_{x})$-min-separator in $G_{x,T}$ if and only if $Y\cup Z_{x,T}$ is a $(x,T_{x})$-min-separator in $G$, or
\item certifies that $ T_{x} = \emptyset$ or that there is no $k$-scratch $(L,S,R)$ where $\emptyset \neq T_{x}\subseteq R$,
$\ltil\in[|L|/2,|L|]$ and $x\in L$.
\end{itemize}
\end{lem}

Below, we prove the main result of this section using the key lemma (\Cref{lem:kernel}) above. 

\paragraph{Proof of \Cref{lem:main k scratch}.}
For each $i=1,\dots,\lg(k/(100\log n))$, let $\ltil^{(i)}=2^{i}$.
Let $T^{(i,1)},\dots,T^{(i,O(\log n))}$ be independently obtained
by sampling each vertex with probability $1/(8\ltil^{(i)})$ and let
$X^{(i)}$ be a set of $O(n\log n/\ltil^{(i)})$ random vertices. We
invoke \Cref{lem:kernel} with parameters $(\ltil^{(i)},X^{(i)},T^{(i,j)})$
for each $j=1,\dots,O(\log n)$. 
For each $x\in X^{(i)}$ where the kernel graph $G_{x,T^{(i,j)}}$ is returned, we find $(x,t_{x})$-min-separator in $G_{x,T^{(i,j)}}$
by calling the maxflow subroutine and obtain a $(x,T_{x}^{(i,j)})$-min-separator
in $G$ by combining it with $Z_{x,T^{(i,j)}}$. Among all obtained $(x,T_{x}^{(i,j)})$-min-separators
(over all $i,j,x$), we return the one with minimum size as the answer
of \Cref{lem:main k scratch}.
Before returning such cut, we verify in $O(m)$ time that it is indeed a vertex cut in $G$ (as \Cref{lem:kernel} is only correct whp.). If not, we return an arbitrary vertex
cut of $G$ (e.g., $N_G(v)$ where $v$ is a minimum degree vertex). Also, if there is no graph $G_{x,T^{(i,j)}}$ returned
from \Cref{lem:kernel} at all, then we return an arbitrary vertex
cut of $G$ as well.

For correctness, it is clear that the algorithm always returns some vertex cut of $G$ with certainty. Now, suppose that $G$ has a $k$-scratch $(L,S,R)$.
Consider $i$ such that $\ltil^{(i)}\in[|L|/2,|L|]$. Then, there
exists $x\in X^{(i)}$ where $x\in L$ whp. Also, by \Cref{prop:avoid SL},
there is $j$ where $\emptyset \neq T^{(i,j)}_x\subseteq R$ whp. Therefore, a $(x,T_{x}^{(i,j)})$-min-separator
must have size less than $k$ and we must obtain it by \Cref{lem:kernel}. 

Finally, we bound the running time. As we call \Cref{lem:kernel} $O(\log^{2}n)$
times, this takes $\ot(\sum_i (m+|X_i|k\ltil^{(i)})) = \ot(m + \sum_i \frac{n}{\ltil^{(i)}} k\ltil^{(i)}) = \ot(m)$ time. 
The total size of maxflow instances is at most $\sum_{i,j}\sum_{x\in X^{(i)}}|E(G_{x,T^{(i,j)}})|=\sum_{i}\ot(k\ltil^{(i)}\cdot(n/\ltil^{(i)}))=\ot(m)$.
This completes the proof.

\paragraph{Organization of this section.}
We formally
show the existence of $G_{x,T}$ in \Cref{sec:Exist Kernel} (using
the help of reduction rules shown in \Cref{sec:reduc rule}). Next, we give efficient data structures for efficiently building each $G_{x,T}$ in \Cref{sec:Fast Kernel} and then use them to finally prove \Cref{lem:kernel} in \Cref{sec:kernel proof}.

\subsection{Reduction Rules for $(s,t)$-vertex Mincut}
\label{sec:reduc rule}

In this section, we describe a simple and generic ``reduction rules''
for reducing the instance size of the $(s,t)$-vertex mincut problem.
We will apply these rules in \Cref{sec:Exist Kernel}. Let $H=(V,E)$
be an arbitrary simple graph with source $s$ and sink $t$ where $(s,t)\notin E$.

The first rule helps us identify vertices that must be in every mincut and hence we can remove them. More specifically, we can always remove common neighbors of both source $s$ and sink $t$ and work on the smaller graph.
\begin{prop}
[Identify rule]\label{prop:rule identify}Let $H'=H\setminus N(s)\cap N(t)$.
Then, $S'$ is an $(s,t)$-min-separator in $H'$ iff $S=S'\cup(N(s)\cap N(t))$
is an $(s,t)$-min-separator in $H'$.
\end{prop}

\begin{proof}
Let $v\in N(s)\cap N(t)$. Observe that $v$ is contained in \emph{every}
$(s,t)$-separator in $H$. So $S'$ is an $(s,t)$-min-separator
in $H\setminus\{v\}$ iff $S'\cup\{v\}$ is an $(s,t)$-min-separator
in $H$. The claim follows by applying the same argument on another
vertex $v'\in N(s)\cap N(t)\setminus\{v\}$ in $H\setminus\{v\}$
and repeating for all vertices in $N(s)\cap N(t)$.
\end{proof}
The second rule helps us ``filter'' useless edges and vertices w.r.t.
$(s,t)$-vertex connectivity.
\begin{prop}
[Filter rule]\label{prop:rule filter}
There exists a maximum set of  $(s,t)$-vertex-disjoint paths $P_1,\dots,P_z$ in $H$ such that no path contains edges/vertices
that satisfies any of the following properties.
\begin{enumerate}[noitemsep,nolistsep]
\item an edge $e$ with both endpoints in $N(s)$ or both in $N(t)$.
\item a vertex $v$ where $t\in N(v)\subseteq N[t]$.
\item a vertex $v$ where $s$ cannot reach $v$ in $H\setminus N[t]$.
\end{enumerate}
Therefore, by maxflow-mincut theorem, the size of $(s,t)$-vertex mincut in $H$ stays the same
even after we remove these edges and vertices from $H$.
\end{prop}

\begin{proof}
(1): Suppose there exists $P_i = (s,\dots,u_{1},u_{2},\dots,t)$ where $(u_{1},u_{2})\in N(s)\times N(s)$.
We can replace $P_i$ with $P'_i=(s,u_{2},\dots,t)$ which is disjoint
from other paths $P_j$.   The argument is symmetric for $N(t)$.

(2): Let $v$ be such that $t\in N(v)\subseteq N[t]$. We first apply rule (1). This means that $N(v)=\{t\}$.
It is clear that there is no simple $s$-$t$ path through $v$.

(3): Suppose $v\in P_i$. There must exist $t'\in N(t)$ where
$P_i=(s,\dots,t',\dots,v,\dots,t)$ because
$s$ could not reach $v$ if $N[t]$ was removed. Then, we can replace
$P_i$ with $P'_i=(s,\dots,t',t)$ which does not contain $v$ and is still disjoint from other paths $P_j$.
\end{proof}

\subsection{Structure of Kernel $G_{x,T}$}

\label{sec:Exist Kernel}
Let $G$ and $k$ be the input of \Cref{lem:main k scratch}. Throughout
this section, we fix a vertex $x$ and a vertex set $T\neq\emptyset$.
The goal of this section is to show the existence of the graph $G_{x,T}$
as needed in \Cref{lem:kernel} and state its structural properties
which will be used later in \Cref{sec:Fast Kernel,sec:kernel proof}. 

Recall that $T_{x}\defeq T\setminus N[x]$ and also the graph $G'_{x,T}$ is obtained from $G$ by contracting
$T_{x}$ into a \emph{sink} $t_{x}$. We call $x$ a \emph{source}.
Clearly, every $(x,t_{x})$-vertex cut in $G'_{x,T}$ is a $(x,T_{x})$-vertex
cut in $G$.

Let $G_{x,T}$ be obtained from $G'_{x,T}$ by first applying Identify rule from \Cref{prop:rule identify}. Let
$Z_{x,T}=N_{G'_{x,T}}(x)\cap N_{G'_{x,T}}(t_{x})$ be the set removed
from $G'_{x,T}$ by Identify rule. We also write
$Z=Z_{x,T}$ for convenience. After removing $Z_{x,T}$, we apply
Filter rule from \Cref{prop:rule filter}. We call the resulting graph the \emph{kernel} graph $G_{x,T}$.
The reduction rules from Propositions \ref{prop:rule identify} and \ref{prop:rule filter} immediately
imply the following.
\begin{lem}
\label{lem:preserve cut}Any set $Y$ is a $(x,t_{x})$-min-separator
in $G_{x,T}$ iff $Y\cup Z_{x,T}$ is a $(x,T_{x})$-min-separator
in $G$.
\end{lem}

Let us partition vertices of $G_{x,T}$ as follows. Let $N_x=N_{G_{x,T}}(x)$
be the neighborhood of source $x$. Let $N_t=N_{G_{x,T}}(t_{x})$ be the
neighbor of sink $t_{x}$. Note that $N_x$ and $N_t$ are disjoint by Identify rule.
Let $F=V(G_{x,T})\setminus(N_x \cup N_t\cup\{x,t_{x}\})$
be the rest of vertices, which is ``far'' from both $x$ and $t_x$.
By Filter rule(1), $G_{x,T}$ has no internal edges inside 
$N_x$ nor $N_t$. So edges of $G_{x,T}$ can be partitioned to
\begin{equation}
E(G_{x,T})=E_{G_{x,T}}(x,N_x)\cup E_{G_{x,T}}(N_x,F\cup N_t)\cup E_{G_{x,T}}(F,F\cup N_t)\cup E_{G_{x,T}}(N_t,t_{x}).\label{eq:edge of GxT}
\end{equation}
Below, we further characterize each part in $G_{x,T}$ in term of sets
in $G = (V,E)$. See \Cref{fig:kernelization} for illustration.  %

	\begin{figure}[!h] 
\centering
\includegraphics[width=0.6\textwidth ]{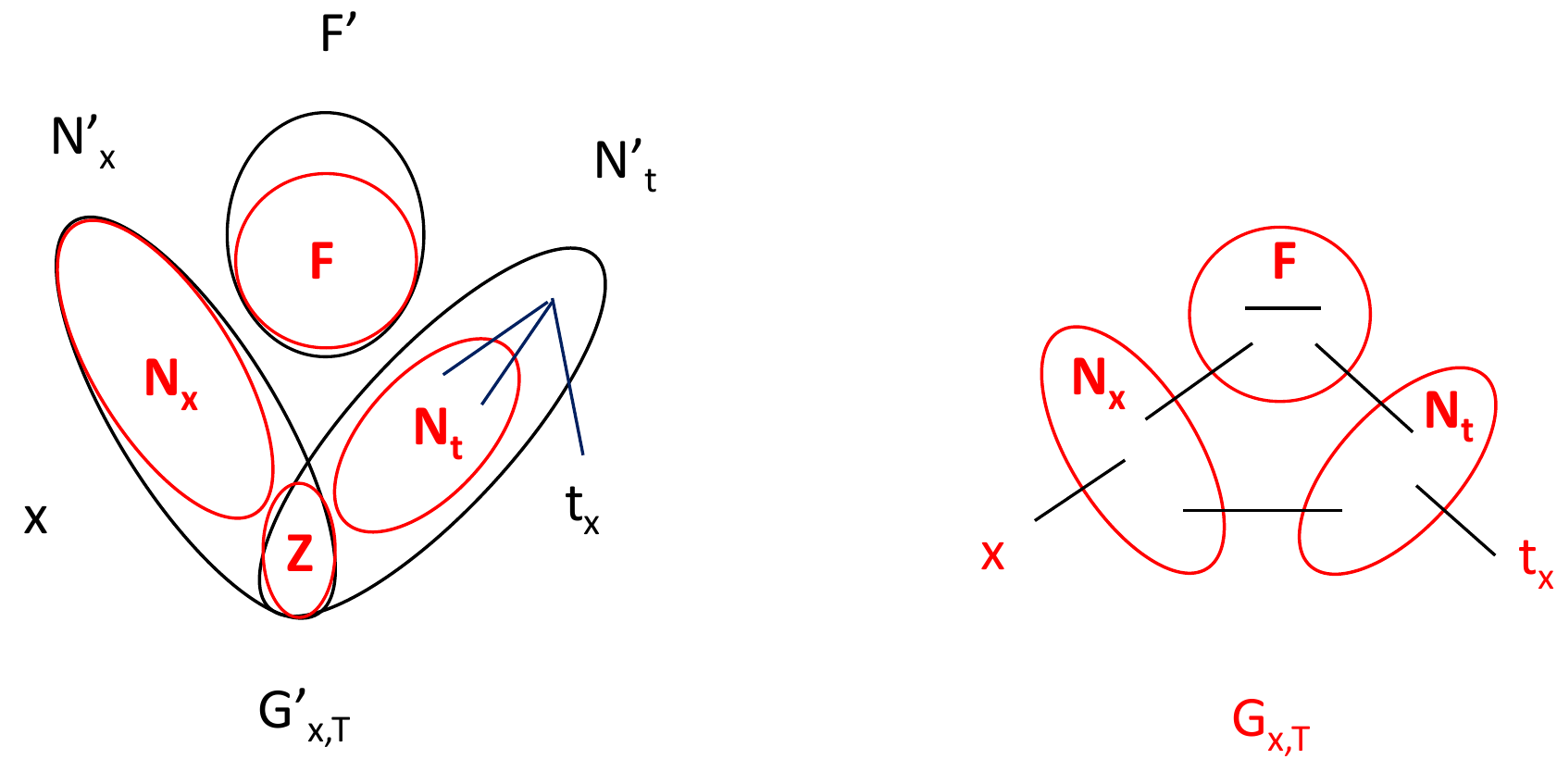} 
\caption{ $G'_{x,T}$ (left) is obtained from $G$ by contracting $T_x$. $G_{x,T}$ (right) is obtained from $G'_{x,T}$ by applying Identify rule and Filter rule, respectively.
The set $Z$ is identified using Identify rule.   
Note that $N_x = N'_x \setminus Z$. The vertices in $F' \setminus F$ cannot be reached from $N_x$ in $G'_{x,T} \setminus N_{G'_{x,T}}(t_x)$. The vertices in $N'_t \setminus (N_t \cup Z)$ have edges only to $N_t$ or $t_x$.
} %
\label{fig:kernelization}
\end{figure}

\begin{lem}
\label{lem:structure GxT}We have the following:
\begin{enumerate}
\item $Z=N(x)\cap N(T_{x})$ and $N_x=N(x)\setminus N(T_{x})$. So, $Z$ and
$N_x$ partition $N(x)$.
\item $F=\{v\in V\setminus(N[x]\cup N[T_{x}])\mid v$ is reachable from
$N_x$ in $G\setminus N[T_{x}]\}$.
\item $N_t=\{v\in N(T_{x})\setminus N[x]\mid v$ is incident to $F$ or $N_x \}$
\end{enumerate}
\end{lem}

\begin{proof}
(1): Observe that $N'_t \defeq  N_{G'_{x,T}}(t_{x}) = N(T_{x})$ and $N'_x \defeq N_{G'_{x,T}}(x) = N(x)$ because $G'_{x,T}$ is simply $G$ after contracting $T_{x}$.
So $Z=N(x)\cap N(T_{x})$. After removing $Z$ from $G'_{x,T}$ via Identify rule, the
remaining neighbor set of $x$ is $N(x)\setminus N(T_{x})$. Since Filter rule never further removes any neighbor
of the source $x$, we have $N_x=N(x)\setminus N(T_{x})$. 

(2): Let $F'=V\setminus(N[x]\cup N[T_{x}])$. Note that $F'$ is precisely the set of vertices in $G'_{x,T}$ that is not dominated by source $x$ or sink $t_{x}$. 
As $F$ is an analogous set for $G_{x,T}$ and $G_{x,T}$ is a subgraph of $G'_{x,T}$, we have $F\subseteq F'$. Observe that only Filter rule(3) may remove vertices from $F'$. (Identify rule and Filter rule(1,2) do not affect $F'$). Now, Filter rule(3)
precisely removes vertices in $F'$ that are not reachable from source
$x$ in $G'_{x,T}\setminus N_{G'_{x,T}}[t_{x}]$. Equivalently, it
removes those that are not reachable from $N_x$ in $G\setminus N[T_{x}]$.
Hence, the remaining part of $F'$ in $G_{x,T}$ is exactly $F$.

(3): Let $N''_t=N(T_{x})\setminus N[x]$. $N''_t$ precisely contains neighbors
of sink $t_{x}$ in $G'_{x,T}$ outside $N[x]$. As $N_t$ is the neighbor set of $t_{x}$ in $G_{x,T}$ and $Z=N(x)\cap N(T_{x})$ is removed
from $G_{x,T}$, we have that $N_t \subseteq N''_t$. Now, only Filter rule(2)
may remove vertices from $N''_t$, and it precisely removes those that
are not incident to $F$ or $N_x$. Therefore, the remaining part of $N''_t$
in $G_{x,T}$ is exactly $N_t$.
\end{proof}
Next, we show we bound the size of $|E(G_{x,T})|$. 
\begin{lem}
\label{lem:small GxT}Suppose \Cref{eq:low out deg} holds. Then, $|E(G_{x,T})|=O((k+|F|)\ltil\log n)$.
\end{lem}

\begin{proof}
We bound $|E(G_{x,T})|$ by bounding each term in \Cref{eq:edge of GxT}.
First, $|E(x,N_x)|\le|N_x|\le|L\cup S|\le2k$. Next, for any $v\in V(G_{x,T})\setminus N_t$,
we have $|E_{G_{x,T}}(v,F\cup N_t)|\le40\ltil\ln n$ by \Cref{eq:low out deg}.
So $|E_{G_{x,T}}(N_x,F\cup N_t)\cup E_{G_{x,T}}(F,F\cup N_t)|\le(|N_x|+|F|)\cdot40\ltil\ln n=O((k+|F|)\ltil\log n)$
because $|N_x|\le2k$. Lastly, each vertex in $N_t$ must have a neighbor in either $N_x$
or $F$ by \Cref{lem:structure GxT}(3). So $|E_{G_{x,T}}(N_t,t_{x})|\le|N_t|\le|E_{G_{x,T}}(N_t,N_x\cup F)|$
which can be charged to either $E_{G_{x,T}}(N_x,F\cup N_t)$ or $E_{G_{x,T}}(F,F\cup N_t)$
whose size are $O((k+|F|)\ltil\log n)$. To conclude, $|E(G_{x,T})|=O((k+|F|)\ltil\log n)$.
\end{proof}
Since $|E(G_{x,T})|$ depends on $|F|$, we will bound $|F|$ as follows. We show that the set $F'_{\rlx}=\{v\in V\setminus N[x]\mid|N(v)\setminus N[x]|\le100\ltil\log n\}$ is a superset of $F$ and then bound $|F'_{\rlx}|$.
The bound on $|F'_{\rlx}|$ will be also used later in \Cref{prop:time BFS} for proving efficiency of our algorithm.
\begin{lem}
\label{lem:small P}Suppose \Cref{eq:low out deg} holds and there
is a $k$-scratch $(L,S,R)$ where $\emptyset\neq T_{x}\subseteq R$,
$\ltil\in[|L|/2,|L|]$, and $x\in L$. Then, we have that $F\subseteq F'_{\rlx}$
and $|F'_{\rlx}|\le16k$. 
\end{lem}

\begin{proof}
Let $F'=V\setminus(N[x]\cup N[T_{x}])$, which is precisely the set
of vertices in $G'_{x,T}$ that is not dominated by source $x$ or
sink $t_{x}$. We have $F\subseteq F'$ because $G_{x,T}$ is a subgraph of $G'_{x,T}$ and we have $F'\subseteq F'_{\rlx}$ because of \Cref{eq:low out deg}.
Now, we bound $|F'_{\rlx}|$. Recall that the definition of a $k$-scratch
$(L,S,R)$ says that $|S_{\low}|\ge300|L|\ln n$ where $|S_{\low}|=\{v\in S\mid\deg(v)\le8k\}$.
Let $N_{\low}=N(x)\cap S_{\low}$. We will show that $|F'_{\rlx}|\cdot|S_{\low}|/2\le|E_{G}(N_{\low},F'_{\rlx})|\le 8k|S_{\low}|.$
This would imply that $|F'_{\rlx}|\le16k$ and complete the proof
of the claim. 

The upper bound on $|E_{G}(N_{\low},F'_{\rlx})|$ follows because
$N_{\low}\subseteq S_{\low}$ and each vertex in $S_{\low}$ has degree
at most $8k$. To prove the lower bound on $|E_{G}(N_{\low},F'_{\rlx})|$,
we will actually show that for every $v\in F'_{\rlx}$, $|E_{G}(v,N_{\low})|\ge|S_{\low}|/2$.
We have
\begin{align*}
|E_{G}(v,N_{\low})| & \ge|E_{G}(v,N(x))|-|N(x)\setminus N_{\low}|\\
 & \ge k-100\ltil\ln n-(k+|L|-|N_{\low}|)\\
 & =|N_{\low}|-100\ltil\ln n-|L|\\
 & \ge|S_{\low}|-|L|-100\ltil\ln n-|L|\\
 & \ge|S_{\low}|-102|L|\ln n\ge|S_{\low}|/2.
\end{align*}
To see the second inequality, we have $|E_{G}(v,N(x))|\ge k-100\ltil\ln n$
because $\deg_{G}(v)\ge k$ but $|E_{G}(v,V\setminus N(x))|\le100\ltil\ln n$
by definition of $F'_{\rlx}$. Also, $|N(x)|\le|L\cup S|\le k+|L|$.
The third inequality follows because $|S_{\low}|\le|N_{\low}|+|S_{\low}\setminus N_{\low}|\le|N_{\low}|+|L|$
by \Cref{prop:outside Nx} (the part of $S$ outside $N[x]$ has size
less than $L$, and so the part of $S_{\low}$ outside $N_{\low}$
has size less than $L$ as well). This completes the proof of the
claim.
\end{proof}

\subsection{Data Structures}

\label{sec:Fast Kernel}

In this section, we show fast data structures needed for proving \Cref{lem:kernel}.
Throughout this section, let $(G,k,\ltil,T)$ denote the input given
to \Cref{lem:kernel}. We treat them as global variables in this
section. Moreover, as the guarantee from \Cref{eq:low out deg} holds
whp, \textbf{we will assume that \Cref{eq:low out deg} holds in this section.}

There are two steps. First, we build an oracle that, given any vertices
$x$ and $v$, lists all neighbors of $v$ outside $N[x]$ if the
set is small. Second, given an arbitrary vertex $x$, we use this
oracle to perform a BFS-like process that allows us to gradually build
$G_{x,T}$ without having an explicit representation of $G_{x,T}$
in the beginning. We show how to solves these tasks respectively
in the subsections below. 

\subsubsection{An Oracle for Listing Neighbors Outside $N[x]$}

In this section, we show the following data structure.
\begin{lem}
[Neighbor Oracle]\label{lem:list neighbor}There is an algorithm
that preprocesses $(G=(V,E),k,\ltil)$ in $\ot(m)$ time and supports
queries $\outneigh(x,v)$ for any vertex $x$ where $|N[x]|\le k+2\ltil$
and $v\in V\setminus\{x\}$.

$\outneigh(x,v)$ either returns the neighbor set of $v$ outside
$N[x]$, i.e.~$N(v)\setminus N[x]$, in $\ot(\ltil)$ time, or report
``too big'' in $O(\log n)$ time. If $|N(v)\setminus N[x]|\le40\ltil\ln n$, then $N(v)\setminus N[x]$
is returned. If $|N(v)\setminus N[x]|>100\ltil\ln n$, then ``too
big'' is reported. Whp, every query is answered correctly.
\end{lem}

For any vertex set $V'\subseteq V$, let the indicator vector $\one_{V'}\in\{0,1\}^{V}$
of $V'$ be the vector where $\one_{V'}(u)=1$ iff $u\in V'$. In this
section, we always use sparse representation of vectors, i.e.~a list
of (index,value) of non-zero entries of the vector. 

The algorithm preprocesses as follows. Set $s\gets100\ltil\ln n$.
For every vertex $v\in V$, we compute the sketches $\sketch_{s}(\one_{N(v)})$,
$\sketch_{s}(\one_{N[v]})$, $\sketch_{\ltwo}(\one_{N(v)})$, and
$\sketch_{\ltwo}(\one_{N[v]})$ using Theorems  \ref{thm:estimate} and \ref{thm:recovery}.
Observe the following.
\begin{prop}
The preprocessing time is $\ot(m)$.
\end{prop}

\begin{proof}
Theorems \ref{thm:estimate} and \ref{thm:recovery} preprocess in  $\ot(n)$ time. The total time to compute the sketches is $\sum_{v\in V}\ot(\deg(v))=\ot(m)$.
\end{proof}
Now, given a vertex $x$ where $|N[x]|\le k+2\ltil$ and $v\in V\setminus\{x\}$,
observe that the non-zero entries of $\one_{N(v)}-\one_{N[x]}\in\{-1,0,1\}^{V}$
corresponds to the symmetric difference $(N(v)\setminus N[x])\cup(N[x]\setminus N(v))$.
We will bound the size of $N[x]\setminus N(v)$ as follows:
\begin{align}\label{eq:sym}
|N[x]\setminus N(v)|\le|N[x]|-|N[x]\cap N(v)|\le k+2\ltil-(k-|N(v)\setminus N[x]|)=|N(v)\setminus N[x]|+2\ltil
\end{align}
where the second inequality is because $k\le|N(v)|=|N(v)\cap N[x]|+|N(v)\setminus N[x]|$.
Therefore, we have
\[
|N(v)\setminus N[x]|\le\|\one_{N(v)}-\one_{N[x]}\|_{0} \overset{(\ref{eq:sym})}{\le}2|N(v)\setminus N[x]|+2\ltil.
\]
Since  $\one_{N(v)}-\one_{N[x]}\in\{-1,0,1\}^{V}$,   we have 
\begin{align} \label{eq:non-zero} |N(v)\setminus N[x]|\le\|\one_{N(v)}-\one_{N[x]}\|_{2} \le2|N(v)\setminus N[x]|+2\ltil.
\end{align} 

Now, we describe how to answer the query. First, we compute $\sketch_{\ltwo}(\one_{N(v)})-\sketch_{\ltwo}(\one_{N[x]})=\sketch_{\ltwo}(\one_{N(v)}-\one_{N[x]})$
in $O(\log n)$ time. If $\|\sketch_{\ltwo}(\one_{N(v)}-\one_{N[x]})\|_{2}>s$,
then we report ``too big''. Otherwise, we have $s\ge\|\sketch_{\ltwo}(\one_{N(v)}-\one_{N[x]})\|_{2}\ge\|\one_{N(v)}-\one_{N[x]}\|_{2}=\|\one_{N(v)}-\one_{N[x]}\|_{0}$
by \Cref{thm:estimate} and because $\one_{N(v)}-\one_{N[x]}\in\{-1,0,1\}^{V}$.
So, we can compute $\sketch_{s}(\one_{N(v)})-\sketch_{s}(\one_{N[x]})=\sketch_{s}(\one_{N(v)}-\one_{N[x]})$
and obtain the set $N(v)\setminus N[x]$ inside $(N(v)\setminus N[x])\cup(N[x]\setminus N(v))$
 using \Cref{thm:recovery} in $\ot(s)=\ot(\ltil)$ time.

To see the correctness, if $|N(v)\setminus N[x]|\le40\ltil\ln n$,
then 
\[
\|\sketch_{\ltwo}(\one_{N(v)}-\one_{N[x]})\|_{2}\le1.1\|\one_{N(v)}-\one_{N[x]}\|_{2} \overset{(\ref{eq:non-zero})}{\le}1.1\cdot(2|N(v)\setminus N[x]|+2\ltil)\le100\ltil\ln n=s.
\]
So the set $N(v)\setminus N[x]$ must be returned. If $|N(v)\setminus N[x]|>100\ltil\ln n$,  
then $$\|\sketch_{\ltwo}(\one_{N(v)}-\one_{N[x]})\|_{2}\geq \|\one_{N(v)}-\one_{N[x]}\|_{2} \overset{(\ref{eq:non-zero})}{\ge}|N(v)\setminus N[x]|>100\ltil\ln n.$$
and so ``too big'' is reported in $O(\log n)$ time.
Every query
is correct whp because of the whp guarantees from \Cref{thm:estimate,thm:recovery}. This completes the proof of \Cref{lem:list neighbor}.

\subsubsection{Building $G_{x,T}$ by Sketchy Search}

In this section, we show how to use the oracle from \Cref{lem:list neighbor}
to return the kernel graph $G_{x,T}$. 
As the oracle is based on linear sketching and we use it in a BFS-like process, we call this algorithm \emph{sketchy search}.
\begin{lem}
[Sketchy Search]\label{lem:BFS}There is an algorithm that preprocesses
$(G,k,\ltil,T)$ in $O(m)$ time and guarantees the following whp.

Given a query vertex $x\in X$, by calling the oracle from \Cref{lem:list neighbor}, return either $\bot$ or the kernel
graph $G_{x,T}$ with $O(k\ltil\log n)$ edges together with the set
$Z_{x,T}$  (defined in the beginning of \Cref{sec:Exist Kernel}) in
$\ot(k\ltil)$ time.
If $T_{x}\neq\emptyset$ and there is a $k$-scratch $(L,S,R)$ where
$T_{x}\subseteq R$, $\ltil\in[|L|/2,|L|]$, and $x\in L$,
then the algorithm must return $G_{x,T}$ and $Z_{x,T}$. 
\end{lem}

The remaining part of this section is for proving \Cref{lem:BFS}. 
In the preprocessing step, we just compute $V_{\bad}=\{v\mid
T\subseteq N[v]\}$ by trivially checking if $T\subseteq N[v]$ on each
vertex $v$ using total $\sum_{v}\deg(v)=O(m)$ time. Observe that $x\in
V_{\bad}$ iff $T_{x}=\emptyset$.

  Next, if there is a $k$-scratch $(L,S,R)$ where $x\in L$ and $\ltil\in[|L|/2,|L|]$, then we must have $|N[x]|\le k+|L|\le k+2\ltil$. So, given a query vertex $x$, if $x\in V_{\bad}$ or $|N[x]|>k+2\ltil$, we can just return $\bot$. From now, we assume that $T_{x}\neq\emptyset$
and $|N[x]|\le k+2\ltil$. %

Before showing how to construct $G_{x,T}$, we recall the definitions
of $G_{x,T}$ and $Z_{x,T}$ from \Cref{sec:Exist Kernel}. \Cref{eq:edge of GxT}
says that edges of $G_{x,T}$ can be partitioned as
$$E(G_{x,T})=E_{G_{x,T}}(x,N_x)\cup E_{G_{x,T}}(N_x,F\cup N_t)\cup E_{G_{x,T}}(F,F\cup N_t)\cup E_{G_{x,T}}(N_t,t_{x})$$
where $N_x=N_{G_{x,T}}(x)$, $N_t=N_{G_{x,T}}(t_{x})$, and $F=V(G_{x,T})\setminus(N_x\cup N_t\cup\{x,t_{x}\})$.
We write $Z\defeq Z_{x,T}=N_{G'_{x,T}}(x)\cap N_{G'_{x,T}}(t_{x})$
where $G'_{x,T}$ is obtained from $G$ by contracting $T_{x}$ into a single vertex $t_{x}$. 

Our strategy is to exploit $\outneigh$ queries from \Cref{lem:list neighbor}
to perform a BFS-like process on $G_{x,T}$  that allows us to gradually
identify $G_{x,T}$ and $Z_{x,T}$ without having an explicit representation
of $G_{x,T}$ in the beginning. The algorithm initializes 
$\Ztil,\Ntil_x,\Etil_{N_x},\Ntil_t,\Ftil,\Etil_{F} =\emptyset$.
At the end of the algorithm, these sets will become
$Z,N_x, E_{G_{x,T}}(N_x,F\cup N_t),N_t,F,E_{G_{x,T}}(F,F\cup N_t)$
respectively. 

Observe that once we know all these sets we can immediately deduce  $E_{G_{x,T}}(x,N_x)$ and $E_{G_{x,T}}(N_t,t_{x})$, and hence we
obtain all parts in $E(G_{x,T})$. So we can return $G_{x,T}$ and $Z_{x,T}$ as desired. 

The algorithm has two main loops. After the first loop, $\Ztil,\Ntil_x$, and $\Etil_{N_x}$
become $Z,N_x$, and $E_{G_{x,T}}(N_x,F\cup N_t)$ respectively. After the second
loop, $\Ntil_t,\Ftil$, and $\Etil_{F}$ become $N_t,F$, and $E_{G_{x,T}}(F,F\cup N_t)$
respectively. Let $\countList=0$ initially. We use $\countList$
to count the number of times that $\outneigh(x,v)$ lists neighbors
of $v$ (not just reports ``too big''). In \Cref{alg:build GxT}, we describe this BFS-like
process in details.   %

\begin{algorithm}
\begin{enumerate}
\item \label{enu:forloop}For each $v\in N(x)$, 
\begin{enumerate}
\item Set $\textsc{visit}(v)=\textsc{true}$.
\item \label{enu:Z toobig}If $\outneigh(x,v)$ returns ``too big'', then
add $v$ to $\Ztil$.
\item Else,  $\outneigh(x,v)$ returns the set $N(v)\setminus N[x]$. 
\begin{enumerate}
\item \label{enu:Z incident}If $N(v)\setminus N[x]$ intersects $T_{x}$,
then add $v$ to $\Ztil$. 
\item \label{enu:N}Else, (1) add $v$ to $\Ntil_x$ and edges between $v$
and $N(v)\setminus N[x]$ to $\Etil_{N_x}$, and (2) add $\{w\in N(v)\setminus N[x]\mid\textsc{visit}(w)\neq\textsc{true}\}$
to $\queue$.
\end{enumerate}
\end{enumerate}
\item \label{enu:whileloop}While $\exists v\in\queue$,
\begin{enumerate}
\item Remove $v$ from $\queue$. Set $\textsc{visit}(v)=\textsc{true}$. 
\item \label{enu:D too big}If $\outneigh(x,v)$ returns ``too big'',
then add $v$ to $\Ntil_t$.
\item Else, $\outneigh(x,v)$ returns the set $N(v)\setminus N[x]$. 
\begin{enumerate}
\item $\countList\gets\countList+1$.
\item \label{enu:D incident}If $N(v)\setminus N[x]$ intersects $T_{x}$,
then add $v$ to $\Ntil_t$.
\item \label{enu:P}Else, (1) add $v$ to $\Ftil$ and edges between $v$
and $N(v)\setminus N[x]$ to $\Etil_{F}$, and (2) add $\{w\in N(v)\setminus N[x]\mid\textsc{visit}(w)\neq\textsc{true}\}$
to $\queue$.
\item \label{enu:big graph}If $\countList>16k$, return $\bot$ and
  terminate. 
\end{enumerate}
\end{enumerate}
\end{enumerate}
\caption{\label{alg:build GxT}An algorithm for building $G_{x,T}$}

\end{algorithm}

Before prove the correctness of \Cref{alg:build GxT}, we observe the
following simple fact.
\begin{fact}
\label{fact:check incident}$N(v)\setminus N[x]$ intersects $T_{x}$
iff $v$ is incident to $T_{x}$.
\end{fact}

\begin{proof}
As $T_{x}\cap N[x]=\emptyset$, we have $N(v)\setminus N[x]$
intersects $T_{x}$ iff $N(v)$ intersects $T_{x}$ iff $v\in N(T_{x})$.
\end{proof}
That is, the condition in Steps \ref{enu:Z incident} and \ref{enu:D incident}
is equivalent to checking if $v$ is incident to $T_{x}$. %
Now, we prove the correctness of the first loop.
\begin{prop}
\label{prop:correct for}After the for loop in Step \ref{enu:forloop},
$\Ztil$, $\Ntil_x$, and $\Etil_{N_x}$ become $Z$, $N_x$, and $E_{G_{x,T}}(N_x,F\cup N_t)$,
respectively.
\end{prop}

\begin{proof}
By \Cref{lem:structure GxT}(1), $N(x)=Z\dot{\cup}N_x$ where $Z=N(x)\cap N(T_{x})$
and $N_x=N(x)\setminus N(T_{x})$. After the for loop, every $v\in N(x)$
is added to either $\Ztil$ or $\Ntil_x$. If $v$ is added to $\Ztil$
in Step \ref{enu:Z incident}, then \Cref{lem:list neighbor} implies that $|N(v)\setminus N[x]|>40\ltil\ln n$
and so $v\in N(T_{x})$ by \Cref{eq:low out deg}, which means $v\in Z$.
If $v$ is added to $\Ztil$ in Step \ref{enu:Z incident}, then we
directly verify that $v\in N(T_{x})$ (see \Cref{fact:check incident})
and so $v\in Z$ again. Lastly, if $v$ is added to $\Ntil_x$ in Step
\Cref{enu:N}, then $v\notin N(T_{x})$ and so $v\in N_x$. This means
that indeed $\Ztil=Z$ and $\Ntil_x=N_x$ after the for loop. Lastly,
every time $v$ is added to $\Ntil_x$, we add $E_{G}(v,V\setminus N[x])=E_{G_{x,T}}(v,F\cup N_x)$
into $\Etil_{N_x}$. So $\Etil_{N_x}$ also collects all edges in $E_{G_{x,T}}(N_x,F\cup N_t)$
after the for loop.
\end{proof}
Next, we prove the correctness of the second loop. The proof is similar to the first one
but more complicated.

\begin{prop}
\label{prop:correct while}Suppose $\bot$ is not returned by \Cref{alg:build GxT}.
Then, at end of the while loop in Step \ref{enu:whileloop}, $\Ntil_t$,
$\Ftil$ and $\Etil_{F}$ become $N_t$, $F$ and $E_{G_{x,T}}(F,F\cup N_t)$
respectively.
\end{prop}

\begin{proof}
We will prove by induction on time that (1) $\Ntil_t \subseteq N_t$, (2)
$\Ftil\subseteq F$, and (3) if $w\in\textsc{Queue}$ at some point
of time, then $w\in N_t\cup F$.

For the base case, consider the time before the while loop is executed.
We have $\Ntil_t=\emptyset$ and $\Ftil=\emptyset$. If $w\in\queue$,
then $w\in N(v)\setminus N[x]$ for some $v\in N_x$. There are two
cases: if $w\in N(T_{x})$, then $w\in N(T_{x})\setminus N[x]$ and
$w$ is incident to $v\in N_x$, which means that $w\in N_t$ by \Cref{lem:structure GxT}(3).
Otherwise, if $w\notin N(T_{x})$, then $w\notin N[x]\cup N(T_{x})$
and $(v,w)$ is a path from $N_x$ to $w$ in $G\setminus N[T_{x}]$,
which means that $w\in F$ by \Cref{lem:structure GxT}(2). 

For the inductive step, consider that iteration where we visit $v$.
We prove the three statements below one by one. 
\begin{enumerate}
	\item Suppose $v$
	is added to $\Ntil_t$. If $v$ is added at Step \ref{enu:D too big},
	then \Cref{lem:list neighbor} implies that $|N(v)\setminus N[x]|>40\ltil\ln n$ and so $v\in N(T_{x})$
	by \Cref{eq:low out deg}. If $v$ is added at Step \ref{enu:D incident},
	then we directly verify that $v\in N(T_{x})$ (see \Cref{fact:check incident}).
	In both cases, $v\in N(T_{x})$. As $v\in N_t\cup F$ by induction,
	$v$ must be in $N_t$. So $\Ntil_t\subseteq N_t$ holds.
	\item  Suppose 	$v$ is added to $\Ftil$, which only happens at Step \ref{enu:P}.
	We directly verify that $v\notin N(T_{x})$. As $v\in N_t\cup F$ by
	induction, $v$ must be in $F$ and so $\Ftil \subseteq F$ holds.
	\item Suppose
	$w$ is added into $\queue$ at Step \ref{enu:P}. There are two cases. If $w\in N(T_{x})$,
	then $w\in N(T_{x})\setminus N[x]$ and $w$ is incident to $v\in F$,
	which means that $w\in N_t$ by \Cref{lem:structure GxT}(3). Otherwise,
	if $w\notin N(T_{x})$, then $w\notin N[x]\cup N(T_{x})$. As $v\in F$,
	there exists a path $p_{v}$ from $N_x$ to $v$ in $G\setminus N[T_{x}]$.
	Now, observe that the concatenated path $p_{w}=p_{v}\circ(v,w)$
        is a path from
	$N_x$ to $w$ in $G\setminus N[T_{x}]$. So, $w\in F$ by \Cref{lem:structure GxT}(2). In either case, we have $w\in N_t \cup F$.
\end{enumerate}

To show that $\Ntil_t=N_t$ and $\Ftil=F$ at the end, we argue that all
vertices in $N_t \cup F$ must be visited at some point. Observe that our
algorithm simulate a BFS algorithm on $G\setminus Z$ when we start
the search from vertices in $N_x$. Moreover, it never continues the
search once it reaches vertices in $N_t$. By \Cref{lem:structure GxT}(2), vertices
in $F$ are reachable from $N_x$ in $G\setminus N[T_{x}]\subseteq G\setminus Z$.
So all vertices from $F$ must be visited. Also, because $N_t \subseteq N(T_{x})\setminus N[x]$
and every vertex in $N_t$ is incident to $F$ or $N_x$, all vertices
from $N_t$ must be visited as well. This completes the proof that $\Ntil_t=N_t$
and $\Ftil=F$ at the end of the while loop. 

Finally, every time $v$ is added to $\Ftil$, we add $E_{G}(v,V\setminus N[x])=E_{G_{x,T}}(v,F\cup N_t)$
into $\Etil_{F}$. So $\Etil_{F}$ collects all edges in $E_{G_{x,T}}(F,F\cup N_t)$
after the while loop.
\end{proof}
Let $v$ be a visited vertex in some iteration of the for loop or
the while loop. We say that $v$'s iteration is \emph{fast} if $\outneigh(x,v)$
returns ``too big'', otherwise we say that $v$'s iteration is
\emph{slow}.

\begin{prop}
\label{prop:time BFS}\Cref{alg:build GxT} takes $\ot(k\ltil)$ time.
\end{prop}

\begin{proof}

By \Cref{lem:list neighbor}, each fast iteration takes $O(\log n)$
time. 
For each slow iteration, the bottle necks are (i) listing vertices in $N(v) \setminus N[x]$, and (ii) checking if $N(v) \setminus N[x]$ intersects $T_x$.
The former takes $\ot(\ltil)$ time by \Cref{lem:list neighbor}. The latter also takes $|N(v) \setminus N[x]| = O(\ltil)$ time because we can simply check, for every $w \in N(v) \setminus N[x]$, if $w \in T_x$ which happens iff $w \in T$. 

Observe the number of slow iterations
is at most $|N(x)|+\countList\le2k+16k$ by the condition in \Cref{enu:big graph}.
So the total time on slow iterations is at most $\ot(k\ltil)$. We
claim the number of fast iterations is most $\countList\cdot100\ltil\ln n$,
which would imply that the total running time  is $\ot(k\ltil)$.

To prove that claim, we say that $w$ is a \emph{child} of $v$ if
$w$ is added to $\queue$ at $v$'s iteration. If $v$'s iteration
is fast, then $v\in \Ztil\cup \Ntil_t$ and so $v$ has no child. If $v$'s
iteration is slow, then $v$ has at most $|N(v)\setminus N[x]|\le100\ltil\ln n$
children by \Cref{lem:list neighbor}. This implies that there are
at most $\countList\cdot100\ltil\ln n$ fast iterations as desired.
\end{proof}
\begin{prop}\label{prop:bot only if bad}
If \Cref{alg:build GxT} returns $\bot$, then there is no $k$-scratch
$(L,S,R)$ where $\emptyset\neq T_{x}\subseteq R$, $\ltil\in[|L|/2,|L|]$,
and $x\in L$.
\end{prop}

\begin{proof}
Recall $F'_{\rlx}=\{v\in V\setminus N[x]\mid|N(v)\setminus N[x]|\le100\ltil\log n\}$
defined above \Cref{lem:small P}.
Observe that if $\countList$ is incremented in $v$'s iteration,
then $|N(v)\setminus N[x]|\le100\ltil\log n$ by \Cref{lem:list neighbor}. As 
 $v\in V\setminus N[x]$, we have $v\in F'_{\rlx}$. 
So, if $\countList>16k$, then $|F'_{\rlx}|>16k$.
As we assume that \Cref{eq:low out deg} holds, \Cref{lem:small P}
implies that there is no $k$-scratch $(L,S,R)$ where $\emptyset\neq T_{x}\subseteq R$,
$\ltil\in[|L|/2,|L|]$, and $x\in L$.
\end{proof}
Now, we conclude with the proof of \Cref{lem:BFS}.

\paragraph{Proof of \Cref{lem:BFS}.}

Let $(G,k,\ltil,T)$ be given. In the preprocessing step, we  compute
$V_{\bad}$ which takes $O(m)$ time.
Given a query $x \in V$, if $x\in V_{\bad}$ or $|N[x]|>k+2\ltil$,
we return $\bot$ and we are done. %
  Otherwise, we execute
\Cref{alg:build GxT} which takes $\ot(k\ltil)$ time by \Cref{prop:time
  BFS}. 
  The algorithm either returns
$\bot$ and otherwise correctly constructs all parts of $G_{x,T}$ by Propositions \ref{prop:correct for} and  \ref{prop:correct while} whp. Using these sets, we can build $G_{x,T}$
via \Cref{eq:edge of GxT} and obtain $Z_{x,T}=Z$ in $\ot(k\ltil)$ time. %
Note that $|F|=|\Ftil|\le16k$ by the condition in Step \ref{enu:big graph}.
So $|E(G_{x,T})|=O(k\ltil\log n)$ by \Cref{lem:small GxT}. 

Finally, if $T_{x}\neq\emptyset$ and there is a $k$-scratch $(L,S,R)$
where $\emptyset\neq T_{x}\subseteq R$, $\ltil\in[|L|/2,|L|]$, and
$x\in L$, then we have $x\notin V_{\bad}$ and $N[x]\le k+2\ltil$,
so $\bot$ is not returned before running \Cref{alg:build GxT}. By \Cref{prop:bot only if bad},
\Cref{alg:build GxT} cannot return $\bot$ as well. So $G_{x,T}$
and $Z_{x,T}$ must be returned.

\subsection{Proof of \Cref{lem:kernel} (Sublinear-time Kernelization)}

\label{sec:kernel proof}

Let $(G,k,\ltil,T,X)$ be given as input. We first initialize the oracle from \Cref{lem:list neighbor} and the BFS-like process from \Cref{lem:BFS}. This takes $\ot(m)$ time. For each
$x\in X$, we query $x$ to the algorithm from \Cref{lem:BFS}. \Cref{lem:BFS}
guarantees that each query takes $\ot(k\ltil)$ time and returns either
$\bot$ or $(G_{x,T},Z_{x,T})$. Therefore, the total running time
is $\ot(m+|X|k\ltil)$. 

For each query $x\in X$, \Cref{eq:low out deg} holds whp by \Cref{prop:low deg whp}. So we will
assume it and conclude the following whp. By \Cref{lem:BFS}, if $\bot$
is returned, then we can correctly certify that $T_{x}=\emptyset$
or there is no $k$-scratch $(L,S,R)$ where $\emptyset\neq T_{x}\subseteq R$,
$\ltil\in[|L|/2,|L|]$, and $x\in L$. If $(G_{x,T},Z_{x,T})$ is
returned, then we have that $|E(G_{x,T})|=O(k\ltil\log n)$. %
By \Cref{lem:preserve cut}, any set $Y$ is a $(x,t_{x})$-min-separator in $G_{x,T}$ iff
$Y\cup Z_{x,T}$ is a $(x,T_{x})$-min-separator in $G$.
as desired.

        \section{Using Isolating Cuts Lemma}\label{sec:using isolating}

We say that a vertex cut $(L,S,R)$ is a \emph{$k$-non-scratch} if it has size less than $k$ but it is not a $k$-scratch. That is,
$(L,S,R)$ is such that (1) $|S|<k$ and $|L|>k/100\log n$, or (2)
$|S|<k$, $|L|\le k/100\log n$, and $|S_{\low}| < 300|L|\ln n$. Recall
that $S_{\low}=S\cap V_{\low}$ and $V_{\low}=\{v\mid\deg(v)\le8k\}$.
In the previous section, we can report that a mincut has size less
than $k$ if a graph contains a $k$-scratch. In this section, we
solves the opposite case; we will report that a mincut has size less
than $k$ if a graph contains a $k$-non-scratch.
More formally, we prove the following.

\begin{lem}
\label{lem:main k nonscratch}There is an algorithm that, given an
undirected graph $G$ with $n$ vertices and $m$ edges and a parameter
$k$ where $m\le nk$, returns a vertex cut $(L,S,R)$ in $G$. If
$G$ has a $k$-non-scratch, then $|S|<k$ w.h.p. The algorithm makes
$s$-$t$ maxflow calls on unit-vertex-capacity graphs with $O(m\log^5 n)$
total number of vertices and edges and takes $\ot(m)$ additional
time.
\end{lem}

Note that the lemma above only applies on graphs with at most $nk$ edges, but
we can easily and will ensure this when we use the lemma in \Cref{sec:together}.
The rest of this section is for proving \Cref{lem:main k nonscratch}.
The key tool in this section is the \emph{isolating cuts lemma} which was introduced in \cite{LiP20deterministic}. We show how to adapt it for vertex connectivity as follows.
\begin{lem}
[Isolating Cuts Lemma] \label{lem:isolator}There exists an algorithm
that takes as inputs $G=(V,E)$ and an independent set $I\subset V$
of size at least $2$, and outputs, for each vertex $v\in I$, a $(v,I\setminus v)$-min-separator
$C_{v}$. The algorithm makes $s$-$t$ maxflow calls on unit-vertex-capacity
graphs with $O(m\log|I|)$ total number of vertices and edges and
takes $O(m)$ additional time.
\end{lem}

We will prove \Cref{lem:isolator} at the end of this section in \Cref{sec:isolating}. Below, we set up the
stage so that we can use it to prove \Cref{lem:main k nonscratch}.
First, we need the following concept:
\begin{definition}
For any vertex set $T$, a vertex cut $(L,S,R)$ \emph{isolates} a
vertex $x$ in $T$ if 
$$L\cap T=\{x\}, S\cap T=\emptyset, \text{and } R\cap T\neq\emptyset.$$\end{definition}

For any $p\in[0,1]$, we let $V(p)$ be obtained by sampling each
vertex in $V$ with probability $p$. Similarly, let $V_{\low}(p)$
be obtained by sampling each vertex in $V_{\low}$ with probability
$p$. The following observation says that, for any a $k$-non-scratch
$(L,S,R)$, we can obtain a random set that $(L,S,R)$ isolates a
vertex in it with good probability.
\begin{prop}
\label{prop:should isolate}Suppose that $G$ has a $k$-non-scratch
$(L,S,R)$. Then, with probability $\Omega(1/\log^{2}n)$, there is
$i\in\{1,\dots,\log n\}$ where $(L,S,R)$ isolates a vertex in $V(\frac{1}{2^{i}})$
or isolates a vertex in $V_{\low}(\frac{1}{2^{i}})$. 
\end{prop}

\begin{proof}
There are two cases. Suppose that $|S|<k$ and $|L|>k/100\log n$.
Consider $p=1/2^{i}$ such that $1\le p(2|L|+|S|)\le2$. As $|S|<100|L|\log n$,
we have $|L|p=\Omega(1/\log n)$. Therefore, $(L,S,R)$ isolates a
vertex in $V(p)$ with probability 
\begin{align*}
 & \Pr[|L\cap V(p)|=1]\cdot\Pr[|S\cap V(p)|=0]\cdot\Pr[|R\cap V(p)|\ge1]\\
 & \ge\Pr[|L\cap V(p)|=1]^{2}\cdot\Pr[|S\cap V(p)|=0]\\
 & =(|L|p\cdot(1-p)^{|L|-1})^{2}(1-p)^{|S|}\\
 & \ge(|L|p)^{2}(1-p)^{2|L|+|S|}=\Omega(1/\log^{2}n)
\end{align*}
where the first inequality is because $|R|\ge|L|$ and the last inequality
follows because $p(2|L|+|S|)=\Theta(1)$ and $|L|p=\Omega(1/\log n)$.

Consider another case where $|S|<k$, $|L|\le k/100\log n$, and $|S_{\low}|\le300|L|\ln n$.
The argument is similar to the previous case, but we first need this claim:
\begin{claim}
\label{claim:helper low}Let $L_{\low}=L\cap V_{\low}$ and $R_{\low}=R\cap V_{\low}$.
We have $L_{\low}=L$ and $|R_{\low}|\ge|L_{\low}|$. 
\end{claim}

\begin{proof}
For each $x\in L$, $N[x]\subseteq L\cup S$. So $\deg(x)\le|L|+|S|\le2k$
and thus $x\in L_{\low}$. So $L_{\low}=L$. To see why $|R_{\low}|\ge|L_{\low}|$,
if $k\ge n/8$, then $V_{\low}=V$ and so $|R_{\low}|=|R|\ge|L|=|L_{\low}|$.
Otherwise, $k<n/8$. So $|L\cup S|\le2k\le n/4$ and then $|R|\ge3n/4$.
As $8k|V\setminus V_{\low}|\le\sum_{v}\deg(v)\le2nk$, we have $|V\setminus V_{\low}|\le n/4$ and so $|V_{\low}|\ge3n/4$.
Therefore, $|R_{\low}|=|R\cap V_{\low}|\ge n/2\ge|L_{\low}|$. 
\end{proof}
Consider $p=1/2^{i}$ such that $1\le p(2|L|+|S_{\low}|)\le2$. As $|S_{\low}|<300|L|\ln n$,
we have $|L|p=\Omega(1/\log n)$. Therefore, $(L,S,R)$ isolates a
vertex in $V_{\low}(p)$ with probability 
\begin{align*}
 & \Pr[|L\cap V_{\low}(p)|=1]\cdot\Pr[|S\cap V_{\low}(p)|=0]\cdot\Pr[|R\cap V_{\low}(p)|\ge1]\\
 & =\Pr[|L_{\low}\cap V(p)|=1]\cdot\Pr[|S_{\low}\cap V(p)|=0]\cdot\Pr[|R_{\low}\cap V(p)|\ge1]\\
 & \ge\Pr[|L\cap V(p)|=1]^{2}\cdot\Pr[|S_{\low}\cap V(p)|=0]\\
 & =(|L|p\cdot(1-p)^{|L|-1})^{2}(1-p)^{|S_{\low}|}\\
 & \ge(|L|p)^{2}(1-p)^{2|L|+|S_{\low}|}=\Omega(1/\log^{2}n)
\end{align*}
where the first inequality by \Cref{claim:helper low} and the last
inequality follows because $p(2|L|+|S_{\low}|)=\Theta(1)$ and $|L|p=\Omega(1/\log n)$.
\end{proof}
The last observation we need is about maximal independent sets of an isolated set.
\begin{prop}
\label{prop:isolate MIS as well}Suppose that a vertex cut $(L,S,R)$
isolates a vertex $x$ in a set $T$. Let $I$ be an maximal independent
set of $T$. Then $(L,S,R)$ also isolates $x$ in $I$.
\end{prop}

\begin{proof}
Note that $|T|\ge2$ because $L\cap T=\{x\}$ and $R\cap T\neq\emptyset$.
As $x$ is not incident to any other vertex in $T$, we have $x\in I$.
So $L\cap I=\{x\}$. Also some vertex in $R\cap T$ must remain in
$I$ because $S\cap T=\emptyset$. So $R\cap I\neq\emptyset$. This
means that $(L,S,R)$ isolates $x$ in $I$.
\end{proof}
Now, we are ready to prove \Cref{lem:main k nonscratch}. 

\paragraph{Proof of \Cref{lem:main k nonscratch}. }

The algorithm for \Cref{lem:main k nonscratch} is as follows. For
each $i\in\{1,\dots,\log n\}$ and $j\in\{1,\dots,O(\log^{3}n)\}$,
we independently sample $T^{(i,j)}=V(\frac{1}{2^{i}})$ and $T_{\low}^{(i,j)}=V_{\low}(\frac{1}{2^{i}})$
and compute maximal independent sets $I^{(i,j)}$ of $T^{(i,j)}$
and $I_{\low}^{(i,j)}$ of $T_{\low}^{(i,j)}$ respectively. Next,
we invoke \Cref{lem:isolator} on $(G,I^{(i,j)})$ if $|I^{(i,j)}|\ge2$
and on $(G,I_{\low}^{(i,j)})$ if $|I_{\low}^{(i,j)}|\ge2$. Among
all separators that \Cref{lem:isolator} returns, we return the one
with minimum size and its corresponding vertex cut. If $|I^{(i,j)}|,|I_{\low}^{(i,j)}|<2$
for all $i,j$, we return an arbitrary vertex cut. 

It is clear the algorithm makes $s$-$t$ maxflow calls on unit-vertex-capacity
graphs with $O(m\log^{5}n)$ total number of vertices and edges and
takes $\ot(m)$ additional time because we invoke \Cref{lem:isolator}
$O(\log^{4}n)$ times. 

To see the correctness, suppose there is a $k$-non-scratch $(L,S,R)$,
then by \Cref{prop:should isolate}, there exist $i$ and $j$ such
that $(L,S,R)$ isolates a vertex in either $T^{(i,j)}$ or $T_{\low}^{(i,j)}$
whp. By \Cref{prop:isolate MIS as well}, $(L,S,R)$ must also isolate
a vertex in either $I^{(i,j)}$ or $I_{\low}^{(i,j)}$ whp. Suppose
that $(L,S,R)$ isolates a vertex $x$ in $I^{(i,j)}$. Then, $(L,S,R)$
is a $(x,I^{(i,j)}\setminus x)$-separator. So the call of \Cref{lem:isolator}
on $(G,I^{(i,j)})$ must return a separator of size at most $|S|<k$.
The argument is the same if $(L,S,R)$ isolates a vertex $x$ in $I_{\low}^{(i,j)}$.

\subsection{Proof of \Cref{lem:isolator} (Isolating Cuts Lemma)}
\label{sec:isolating}

The goal of this section is to prove \Cref{lem:isolator}.
We follow the proof of Theorem~II.2 of \cite{LiP20deterministic}. Order the vertices in $I$ arbitrarily from $1$ to $|I|$, and let the {\em label} of each $v\in I$ be its position in the ordering, a number from $1$ to $|I|$ that is denoted by a unique binary string of length $\lceil\lg|I|\rceil$.
Let us repeat the following procedure for each $i=1,2,\ldots,\lceil\lg|I|\rceil$. Let $A_i\subseteq I$ be the vertices in $I$ whose label's $i$'th bit is $0$, and let $B_i\subseteq I$ be the vertices whose label's $i$'th bit is $1$. Compute a $(A_i,B_i)$-min-separator $C_i\subseteq V$ (for iteration $i$). Note that since $I=A_i\cup B_i$ is an independent set in $G$, the set $V\setminus(A_i\cup B_i)$ is an $(A_i,B_i)$-separator, so an $(A_i,B_i)$-min-separator exists.

First, we show that $G \setminus \bigcup_i C_i$ partitions the set of vertices into connected components each of which contains at most one vertex of $I$. Let $U_v$ be the connected component in $G \setminus \bigcup_i C_i$ containing $v\in I$. Then:
\begin{claim}\label{clm:v}
$U_v\cap I=\{v\}$ for all $v\in I$.
\end{claim}
\begin{proof}
By definition, $v\in U_v\cap I$. Suppose for contradiction that $U_v\cap I$ contains another vertex $u\ne v$. %
Since the binary strings assigned to $u$ and $v$ are distinct, they differ in their $j$'th bit for some $j$. Assume without loss of generality that $u\in A_j$ and $b\in B_j$. Since $C_j\subseteq V$ is a $(A_j,B_j)$-min-separator, there cannot be a $u$-$v$ path whose vertices are disjoint from $C_j$, contradicting the assumption that $u$ and $v$ belong in the same connected component of $G \setminus \bigcup_iC_i$.
\end{proof}

\begin{claim}[Submodularity of vertex cuts]\label{clm:submod}
For any subsets $A,B\subseteq V$, we have
\[ |N(A)| + |N(B)| \ge |N(A\cup B)| + |N(A\cap B)| .\]
\end{claim}
\begin{proof}
We consider the contribution of each vertex $v\in V$ to the LHS $|N(A)|+|N(B)|$ and the RHS $|N(A\cup B)| + |N(A\cap B)|$ separately.
Each vertex $v\in N(A)\cap N(B)$ contributes $2$ to the LHS and at most $2$ to the RHS. Each vertex $v\in N(A)\setminus N(B)$ contributes $1$ to the LHS, and $1$ to the RHS because $v\in N(A\cup B)$ and $v\notin N(A\cap B)$. A symmetric case covers each vertex $v\in N(B)\setminus N(A)$. Finally, each vertex $v\notin N(A)\cup N(B)$ contributes $0$ to both sides.
\end{proof}

Now, for each vertex $v\in I$, let $\lambda_v$ be the size of a $(v,I\setminus v)$-min-separator. For each  $(v,I\setminus v)$-min-separator $C$, we can consider the set $S\subseteq V$ of vertices in the connected component of $G \setminus C$ containing $v$, which necessarily satisfies $N(S)=C$. Let $S^*_v\subseteq V$ be an inclusion-wise minimal set such that $N(S^*_v)$ is a $(v,I\setminus v)$-separator. Then, we claim the following:
\begin{claim}\label{clm:Uv}
$U_v \supseteq S^*_v$ for all $v\in I$.
\end{claim}
\begin{proof}
Fix a vertex $v\in I$ and an iteration $i$. Let $T^i_v\subseteq V$ be the vertices in the connected components of $G \setminus C_i$ that contain at least one vertex the same color as $v$ (on iteration $i$). By construction of $C_i$, the set $T^i_v$ does not contain any vertex of the opposite color. We now claim that $S^*_v\subseteq T_v^i$. Suppose for contradiction that $S^*_v\setminus T_v^i \neq \emptyset$. 
Note that $(S^*_v\cap T_v^i)\cap I=\{v\}$ and 
\[ N(S^*_v\cap T^i_v) \cap I \subseteq (N(S^*_v)\cup N(T^i_v))\cap I \subseteq (N(S^*_v)\cup C_i)\cap I=\emptyset ,\]
where the first inclusion holds because $N(S\cap T) \subseteq N(S) \cup N(T)$ for any $S,T\subseteq V$, and the second inclusion holds because $N(T^i_v)\subseteq C_i$ by construction of $T^i_v$. Therefore,
\[ |N(S^*_v\cap T_v^i)| \ge \lambda_v= |N(S^*_v)|.\]
Indeed, by our choice of $S^*_v$ to be inclusion-wise minimal, we can claim the strict inequality:
\[ |N(S^*_v\cap T_v^i)|>\lambda_v= |N(S^*_v)|. \]
But, by \Cref{clm:submod} we have:
\[  |N( S^*_v\cup T_v^i)| +  |N(S^*_v\cap T_v^i)| \le  |N( S^*_v)| +  |N( T_v^i)|.\]
Therefore, we get:
\[  |N( S^*_v\cup T_v^i)| <  |N( T_v^i)|  .\]
Now observe that $(S^*_v\cup T_v^i)\cap I = T_v^i\cap I$ since $(S^*_v\setminus T_v^i)\cap I = \emptyset$.  In particular, $S^*_v\cup T_v^i$ contains all vertices in $A_i$ and no vertices in $B_i$. Also, since $N(S^*_v)\cap I=\emptyset$ and $N(T^i_v)\cap I=\emptyset$, we also have $N(S^*_v\cup T^i_v)\cap I=\emptyset$. Then, 
\[ |N(S^*_v\cup T^i_v)|<|N(T^i_v)|\le |C_i| ,\]
so $N(S^*_v\cup T^i_v)$ is a smaller $(A_i,B_i)$-separator than $C_i$, a contradiction.

For each iteration $i$, since $S^*_v\subseteq T^i_v$, none of the vertices in $S^*_v$ are present in $C_i$. Note that $G[S^*_v]$ is a connected subgraph; therefore, it is a subgraph of the connected component $U_v$ of $G \setminus \bigcup_i C_i$ containing $v$. This concludes the proof of \Cref{clm:Uv}.
\end{proof}

\begin{fact}\label{fact:flow}
	Given a graph $G=(V,E)$ and distinct vertices $s,t\in V$, and given a $s$-$t$ vertex maxflow, we can compute in $O(|V|+|E|)$ time a set $S\subseteq V$ with $S\cap\{s,t\}=\{s\}$ such that $N(S)$ is a $(s,t)$-min-separator. 
\end{fact}
It remains to compute the desired set $S_v$ given the property that $U_v\supseteq S_v$. 
Construct the graph $G_v$ as follows. Start from the induced graph $G[U_v\cup N_G(U_v)]$, remove all edges with both endpoints in $N_G(U_v)$, and then add a vertex $t$ connected to all vertices in $N_G(U_v)$.  We compute a $v$-$t$ vertex maxflow in $G_v$ and then apply \Cref{fact:flow}, obtaining a set $S_v$ such that $N_{G_v}(S_v)$ is a $(v,t)$-min-separator. Since $t\notin N_{G_v}(S_v)$, we must have $S_v\cap N_{G_v}(U_v)=\emptyset$, so by construction of $G_v$, we have $N_{G_v}(S_v)=N_G(S_v)$. In particular, $N_G(S_v)=N_{G_v}(S_v)\subseteq U_v\cup N_G(U_v)$, and along with $v\in S_v$, we obtain $N_G(S_v)\cap I=\emptyset$.

\Cref{clm:Uv} implies that $N_{G_v}(S^*_v)\subseteq U_v\cup N_{G_v}(U_v)$, so $N_{G_v}(S^*_v)=N_G(S^*_v)$ and $t\notin N_{G_v}(S^*_v)$. Therefore, $N_{G_v}(S^*_v)$ is a $(v,t)$-separator in $G_v$ of size $\lambda_v$. Since $N_{G_v}(S_v)$ is a $(v,t)$-\emph{min}-separator in $G_v$, we have $|N_G(S_v)|=|N_{G_v}(S_v)|\le |N_{G_v}(S^*_v)|=\lambda_v$.  Define $C_v=N(S_v)$, which satisfies the desired properties in the statement of the lemma.

We now bound the total size of the graphs $G_v$ over all $v\in I$.
By construction of the graphs $G_v$, each edge in $E$ joins at most one graph $G_v$. Each graph $G_v$ has $|N_G(U_v)|$ additional edges adjacent to $t$, but since each vertex in $N_G(U_v)$ is adjacent to some vertex in $U_v$ via an edge originally in $E$, we can charge the edges in $G_v$ adjacent to $t$ to the edges originating from $E$. Therefore, the total number of edges over all graphs $G_v$ is $O(m)$. 
Each of the graphs $G_v$ is connected, so the total number of vertices is also $O(m)$. 
Finally, to compute $(A_i,B_i)$-min-separator for all $i$, the total size of the maxflow instances is $O(m\log |I|)$. 
To bound the additional time, by \Cref{fact:flow}, recovering the sets $S_v$ and the values $|N(S_v)|$ takes time linear in the number of edges of $G_v$, which is $O(m)$ time over all $v\in I$.
This completes the proof of \Cref{lem:isolator}.

\section{Putting Everything Together}
\label{sec:together}

For any $k$, we can detect if $G$ has vertex mincut of size less than $k$ as follows. 
First, compute a $k$-connectivity certificate $H$ of $G$ which preserves all vertex cuts of size less than $k$ and $H$ has at most $nk$ edges (so $H$ is applicable for \Cref{lem:main k nonscratch}). This can be done in linear time using the algorithm by Nagamochi and Ibaraki \cite{NagamochiI92}.
Then, we apply Lemmas \ref{lem:main k scratch} and \ref{lem:main k nonscratch} on $H$ with parameter $k$. If $H$ has a vertex cut of size less than $k$, that cut is either a $k$-scratch or $k$-non-scratch, and so one of the algorithms of Lemmas \ref{lem:main k scratch} or \ref{lem:main k nonscratch} must return a vertex cut of size less than $k$ whp.  If vertex mincut of $G$ is at least $k$, then any of the algorithms in \Cref{lem:main k scratch} and \Cref{lem:main k nonscratch} always returns a vertex cut of size at least $k$. 
\Cref{thm:main} follows immediately by a binary search on $k$.

        \section*{Acknowledgement}  
        This project has received funding from the European Research
        Council (ERC) under the European Union's Horizon 2020 research
        and innovation programme under grant agreement No
        715672 and No 759557. Nanongkai was also partially supported by the Swedish
        Research Council (Reg. No. 2019-05622). Panigrahi has been supported in part
        by NSF Awards CCF 1750140 and CCF 1955703.

        \bibliographystyle{alpha}
        \bibliography{references,dp-refs} 

        \appendix

\section{Proofs of Linear Sketching}

\paragraph{Proof of \Cref{thm:estimate}.} 
We can use $F_2$-moment frequency estimation by \cite{AlonMS99}. Although their work focus on estimating on positive entries, their algorithm is linear, and thus it is possible to estimate norm of the difference between two vectors $x,y$:  $\norm{x-y}_2$.%

Given a vector $v \in \mathbb{R}^n$, we compute sketch of $v$ by viewing it in a streaming setting as follows.  We start with a zero vector $x = 0$, and feed a sequence of update $(i,v_i)$ for each non-empty entry in $v$ of total $\norm{v}_0$ updates. Each update can be performed in $\log^{O(1)}(n)$ time.

\paragraph{Proof of \Cref{thm:recovery}.}
The sparse recovery algorithm is described in Section 2.3 in \cite{CormodeF14} (Section 2.3.1 and Section 2.3.2 in particular). In order for their algorithm to work efficiently, we need a standard assumption that the vector $x$ that we compute the sketch from satisfies $x_i \in \mathbb{Z} \cap [-n^{O(1)},n^{O(1)}]$ for all $i \in [n]$ so that all arithmetic operations in this algorithm can be computed in $O(\log n)$ time. %

Given a vector $v \in {\{-1,0,1\}}^n$, we compute sketch of $v$ by viewing it in a streaming setting as follows.  We start with a zero vector $x = 0$, and feed a sequence of update $(i,v_i)$ for each non-zero entry in $v$ of total $\norm{v}_0$ updates. Each update can be performed in $\log^{O(1)}(n)$ time according to their sparse recovery algorithm. 

        \section{Directed Vertex Connectivity}

\label{sec:main directed} The goal of this section is to prove \Cref{thm:main
directed intro}. We first set up our notations on directed graphs. 

\paragraph{Preliminaries.}

Let $G=(V,E)$ be a directed graph. For any set $T$ of vertices,
we let $N_{G}^{\textout}(T)=\{v\notin T\mid\exists u\in T$ and $(u,v)\in E\}$
and $N_{G}^{\textout}[T]=T\cup N_{G}^{\textout}(T)$. Similarly, we
denote $N_{G}^{\textin}(T)=\{v\notin T\mid\exists w\in T$ and $(v,w)\in E\}$
and $N_{G}^{\textin}[T]=T\cup N_{G}^{\textin}(T)$. If $T=\{v\}$,
we also write $N_{G}^{\textout}(v),N^{\textout}[v],N_{G}^{\textin}(v)$
and $N_{G}^{\textin}[v]$. The set $E_{G}(A,B)$ denote the set of
edges $(u,v)$ where $u\in A$ and $v\in B$. If $A=\{v\}$, we write
$E_{G}(v,B)$. We let $\delta_{G}^{\textout}(T)=E_{G}(T,V\setminus T)$
and $\delta_{G}^{\textin}(T)=E_{G}(V\setminus T,T)$.

A \emph{(directed)} \emph{vertex cut} $(L,S,R)$ of a graph $G=(V,E)$
is partition of $V$ such that $L,R\neq\emptyset$ and $E_{G}(L,R)=\emptyset$.
Note again that $E_{G}(L,R)$ is the set of directed edges from $L$
to $R$. We call $S$ the corresponding \emph{(out-)separator} of
$(L,S,R)$. The size of a vertex cut is the size of its separator
$|S|$.%
Let $\kappa_{G}$ denote the size of the directed vertex mincut in
$G$ and we call $\kappa_{G}$ vertex connectivity of $G$.

The \emph{directed vertex connectivity} problem is to find a minimum
vertex cut in a directed graph. In other words, we ask how many vertices
we need to delete so that the resulting graph is not strongly connected.
We show the following directed vertex connectivity algorithm:
\begin{theorem}
\label{thm:main directed}Given an $m$-edge $n$-vertex directed
graph and a parameter $2\leq\ell\leq n/10$, there is a randomized
Monte Carlo vertex connectivity algorithm that runs in$\ot(\frac{n}{\ell}(t_{\textflow}(m,n)+t_{\textflow}(n\ell,n\ell))+t_{\textflow}(m,m))$
time. 
\end{theorem}

The function $t_{\textflow}(m,n)$ above denotes the time to compute
$s,t$-vertex connectivity of an $m$-edge $n$-vertex graph. To get
the above bound, we naturally assume that $x\cdot t_{\textflow}(m,m)\le t_{\textflow}(x\cdot m,x\cdot m)$
for any $x\ge1$. That is, $t_{\textflow}(m,m)$ grows at least linearly
in $m$. Before proving \Cref{thm:main directed}, we show that by plugging in the fastest max flow algorithms, \Cref{thm:main directed intro} follows as a corollary.

\paragraph{Proof of \Cref{thm:main directed intro}:}
Let $\kappa_{G}$ be the
directed vertex connectivity of $G$. We will use the following algorithms
as blackboxes:
\begin{enumerate}
\item \label{enu:alg1}A $\ot(mk^{2})$-time algorithm by \cite{ForsterNYSY20} for computing
vertex connectivity (and its corresponding cut) on any directed graph,
or reporting that $\kappa_{G}>k$.
\item \label{enu:alg2}A $\ot(m+n^{1.5})$-time algorithm by \cite{BrandLNPSS0W20,BrandLLSSW21} for computing
$s$-$t$ max flow on any directed graph.
\item \label{enu:alg3}$A$ $O(m^{4/3+o(1)})$-time algorithm by \cite{KathuriaLS20} for computing $s$-$t$ max flow on any directed unit-capacity graph.
\end{enumerate}
Let $k=n^{0.5-\epsilon}$ and we will set $\epsilon=1/24$ after optimizing
parameters. We first check if $\kappa_{G}\le k$ using \Cref{enu:alg1}
in $\ot(mk^{2})=\ot(mn^{1-2\epsilon})$ time and we assume from now
that $\kappa_{G}>k$. In particular, $m\ge nk=n^{1.5-\epsilon}$.
There are two remaining cases.

First, if $n^{1.5-\epsilon}\le m\le n^{1.5}$, we claim that \Cref{thm:main directed}
implies, by setting $\ell=n^{1/8}$, that there is a vertex connectivity
algorithm with running time 
\begin{align*}
 & \Otil(\frac{n}{\ell}(t_{\textflow}(m,n)+t_{\textflow}(n\ell,n\ell))+t_{\textflow}(m,m))\\
 & =O(\frac{n}{\ell}(n^{1.5}+(n\ell)^{4/3+o(1)})+m^{4/3+o(1)}) & \text{by \Cref{enu:alg2} and \Cref{enu:alg3}}\\
 & =O(n^{19/8+o(1)}) & \text{as }\ell=n^{1/8}\text{ and }m\le n^{1.5}\\
 & =O(mn^{1-2\epsilon+o(1)}) & \text{as }m\ge n^{1.5-\epsilon}\text{ and }\epsilon=\frac{1}{24}
\end{align*}
as desired.

Second, if $m\ge n^{1.5}$, we claim that \Cref{thm:main directed}
implies, by setting $\ell=m^{3/4}/n$, that there is a vertex connectivity
algorithm with running time 
\begin{align*}
 & \Otil(\frac{n}{\ell}(t_{\textflow}(m,n)+t_{\textflow}(n\ell,n\ell))+t_{\textflow}(m,m))\\
 & =O(\frac{n}{\ell}(m+(n\ell)^{4/3+o(1)})+m^{4/3+o(1)}) & \text{by \Cref{enu:alg2} and \Cref{enu:alg3}}\\
 & =O(m^{1/4}n^{2+o(1)}+m^{4/3+o(1)}) & \text{as }\ell=m^{3/4}/n\\
 & =O(mn^{1-2\epsilon+o(1)}) & \text{as }m\ge n^{1.5}\text{ and }\epsilon=\frac{1}{24}.
\end{align*}
In any case, we have obtained a $\ot(mn^{1-1/12+o(1)})$-time algorithm. 

Lastly, if there exists a near-linear time max flow algorithm, then
we have that, by setting $\ell=\Omega(n)$, \Cref{thm:main directed}
implies an algorithm with $\Otil(\frac{n}{\ell}(t_{\textflow}(m,n)+t_{\textflow}(n\ell,n\ell))+t_{\textflow}(m,m))=\Otil(n^{2})$
time.\qed

The rest of this section is for proving \Cref{thm:main directed}.

\paragraph{Proof of \Cref{thm:main directed}:}
By binary search, it suffices to show an algorithm with the following
guarantee. Given a graph $G=(V,E)$ with vertex connectivity $\kappa_{G}$,
the algorithm, with a given parameter $k$, returns some vertex cut
$(L',S',R')$ of $G$ such that, if $\kappa_{G}<k$, then $|S'|<k$
whp. 

Suppose that $G$ has a vertex cut $(L,S,R)$ where $|S|<k$. Our
goal now is to find some cut $(L',S',R')$ of size $|S'|<k$ whp.
We will assume $m\ge nk$, otherwise the problem is trivial because
there is a vertex with degree less than $k$. We also assume w.l.o.g.
that $|L|\le|R|$ by running the algorithm on both $G$ and on the
reverse graph $G^{R}=(V,E^{R})$ where $E^{R}=\{(v,u)\mid(u,v)\in E\}$. 

Let $\ell$ be the parameter given in \Cref{thm:main directed}. There
are three cases. Firstly, we handle the \emph{unbalanced} case when
$|L|\le\ell$ by directly applying the following key lemma in $\Otil(\frac{n}{\ell}\cdot t_{\textflow}(n\ell,n\ell))$
time by setting $a\gets\ell$.
\begin{restatable}[Key Lemma]{lem}{directedKernel}
	\label{lem:main unbalanced directed} There is an algorithm that,
	given an $m$-edge $n$-vertex directed graph $G$ and two parameters
	$a\leq n/10$, and $k\leq n$, returns in $\Otil(\frac{n}{a}\cdot t_{\textflow}(na,na))$
	time a vertex cut $(L',S',R')$. Suppose $G$ contains a vertex cut
	$(L,S,R)$ such that 
	\begin{align}
		|L|\leq a,|S|<k,\mbox{ and }|R|\geq|L|.\label{eq:balanced cut cond}
	\end{align}
	Then, $|S'|<k$ whp.
\end{restatable}

Secondly, we handle the \emph{extreme} case when $k\ge n/2$ and $|L|\le n/10$.
To do this, we just invoke the above lemma to find a vertex cut of
size less than $k$, by setting $a=n/10$, in time $\Otil(t_{\textflow}(n^{2},n^{2}))=\Otil(t_{\textflow}(m,m))$
because $m\ge nk$. So from now, we can assume that 
\begin{equation}
\text{either }k<n/2\text{ or }|L|>n/10.\label{eq:helping}
\end{equation}

Lastly, we can handle the the remaining \emph{balanced} case where
$|L|\ge\ell$ as follows. Independently sample $p=\Otil(n/\ell)$
random pairs $(s_{1},t_{1}),\dots,(s_{p},t_{p})$ of vertices. For
each pair $(s_{i},t_{i})$, if $(s_{i},t_{i})$-vertex connectivity
in $G$ is less than $k$, use the max flow algorithm to return a
vertex cut $(L',S',R')$ where $|S'|<k$. This takes $\Otil(\frac{n}{\ell}(t_{\textflow}(m,n))$
total running time. Observe that if $s_{i}\in L$ and $t_{i}\in R$,
then $(s_{i},t_{i})$-vertex connectivity must be less than $k$ and
we would be done. We claim that this event happens whp:
\begin{prop}
Suppose $|L|\ge\ell$. There exists $1\le i\le p$ where $s_{i}\in L$
and $t_{i}\in R$ whp.
\end{prop}

\begin{proof}
It suffices to prove that, for each $i$, $s_{i}\in L$ and $t_{i}\in R$
with probability $\Omega(\ell/n)$. Given that, there is no such pair
$(s_{i},t_{i})$ with probability at most $(1-\Omega(\frac{\ell}{n}))^{p}\le e^{-\Omega(p\ell/n)}\le n^{-10}$
because $p=\Otil(n/\ell)$. 

If $k\ge n/2$, then $n/10<|L|\le|R|$ by \Cref{eq:helping}. So, for
each $i$, $s_{i}\in L$ and $t_{i}\in R$ with probability $\Omega(1)$.
If $k<n/2$, then $|R|=\Omega(n)$ because $2|R|\ge|L|+|R|\ge n-k\ge n/2$.
So $t_{i}\in R$ with probability $\Omega(1)$ and $s_{i}\in L$ with
probability $|L|/n\ge\ell/n$. So, both happens with probability $\Omega(\ell/n)$.
\end{proof}
By running all three algorithms in time $\Otil(\frac{n}{\ell}(t_{\textflow}(m,n)+t_{\textflow}(n\ell,n\ell))+t_{\textflow}(m,m))$,
if there exists a vertex cut $(L,S,R)$ in $G$ where $|S|<k$, one
of the three algorithms must return a cut $(L',S',R')$ where $|S'|<k$
as desired. This concludes the proof of \Cref{thm:main directed}.
It remains prove \Cref{lem:main unbalanced directed} which we do in
the next section.
        
\section{Using Fast Kernelization for Directed Graphs} \label{sec:main unbalanced directed} 
We prove \Cref{lem:main unbalanced directed} in this section. 

\paragraph{Comparison to the Undirected Case.}
This is the main technical lemma which illustrates that the fast kernelization technique used for proving \Cref{lem:main k scratch} in undirected graphs can be useful in directed graphs too. 
In fact, the proof will follow exactly the same template used for proving \Cref{lem:main k scratch}.
As we need to replace all notations for undirected graphs with the ones for directed graphs, we repeat the whole proof for the ease of verification. 

The main difference is that our technique in directed graphs is not as strong. More specifically, in the \Cref{lem:kernel directed}, we can only bound the size of the kernel graph to be $\Otil(n\ltil)$ instead of $\Otil(k\ltil)$ as in \Cref{lem:kernel}, the analogous lemma for undirected graphs. The running time is also $\Otil(m+|X|n\ltil)$ instead of $\Otil(m+|X|k\ltil)$. As we aim for weaker bounds, the proofs actually simplify a bit more.  

Now, we proceed with the proof of \Cref{lem:main unbalanced directed}.
Throughout this section, we assume that minimum out-degree of $G$ is
at least $k$, otherwise the lemma is trivial. Now, we assume that $G$ contains a vertex cut
$(L,S,R)$ satisfying \Cref{eq:balanced cut cond}. %
We start with a simple observation which says that, given a vertex
$x\in L$, the remaining part of $L\cup S$ outside $N^{\textout}[x]$
has size at most $|L|$.
\begin{prop}
\label{prop:outside Nx directed}For any $x\in L$, $|(L\cup S)\setminus N^{\textout}[x]|<|L|$. 
\end{prop}

\begin{proof}
Note that $N^{\textout}[x]\subseteq L\cup S$ as $x\in L$. The claim follows
because $|L\cup S|<|L|+k$ and $|N^{\textout}[x]|>k$ as the minimum out-degree is
at least $k$.
\end{proof}
We will use $\ltil$ as an estimate of $|L|$ (since $|L|$ is actually
unknown to us). Let $T$ be obtained by sampling each vertex with
probability $1/(8\ltil)$. Let $T_{x}\defeq T\setminus N^{\textout}[x]$ for
any $x\in V$. 
Below, we show two basic properties of $T$. 
\begin{prop}\label{prop:low deg whp directed}
For any $x\in V$, we have the following whp. 
\begin{equation}
\text{For every }v\notin N^{\textin}[T_{x}],\,|N^{\textout}(v)\setminus N^{\textout}[x]|\le40\ltil\ln
n\label{eq:low out deg directed}
\end{equation}
\end{prop}

\begin{proof}
It suffices to prove that, for any $v\in V$, if $|N^{\textout}(v) \setminus N^{\textout}[x]|>40\ltil\ln n$,
then $v$ has an edge to $T_{x}$ whp. Indeed, the probability that 
$v$ does not have an edge to $T_{x}$ is at most $(1-\frac{1}{8\ltil})^{|N^{\textout}(v) \setminus N^{\textout}[x]|} < n^{-5}$. 
\end{proof}
 
\begin{prop}
\label{prop:avoid SL directed}Suppose $|L|/4\le\ltil\le|L|$. For each  $x\in L$,
$\emptyset\neq T_{x}\subseteq R$ with constant probability.
\end{prop}

\begin{proof}
Note that $\emptyset\neq T_{x}\subseteq R$ iff none of vertices from
$(L\cup S)\setminus N^{\textout}[x]$ is sampled to $T$ \emph{and} some vertex
from $R\setminus N^{\textout}[x]$ is sampled to $T$. Observe that $|(L\cup S)\setminus N^{\textout}[x]|<|L|$
by \Cref{prop:outside Nx directed} and $|R\setminus N^{\textout}[x]|=|R|\ge|L|$.

To rephrase the situation, we have two disjoint sets $A_{1}$ and
$A_{2}$ where $|A_{1}|<|L|$ and $|A_{2}|\ge|L|$ and each element
is sampled with probability $\frac{1}{8\ltil}\in[\frac{1}{8|L|},\frac{1}{2|L|}]$.
No element is $A_{1}$ is sampled with probability at least $(1-\frac{1}{2|L|})^{|L|}\ge0.5$.
Some element in $A_{2}$ is sampled with probability at least $1-(1-\frac{1}{8|L|})^{|L|}\ge1-e^{1/8}\ge0.1$.
As both events are independent, so they  happen simultaneously with probability at
least $0.05$. That is, $\emptyset\neq T_{x}\subseteq R$ with probability at least
$0.05$.
\end{proof}

The following key lemma further shows that, given a set $X$, we can build the kernel graph $G_{x,T}$ for each $x\in X$ in $\ot(n\ltil)$ time.
\begin{lem}[Fast Kernelization]
\label{lem:kernel directed}Let $G$ and $k$ be the input of \Cref{lem:main unbalanced directed}.  Let $X$ be a set of vertices. Let
$T$ be obtained by sampling each vertex with probability $1/(8\ltil)$ and $T_{x}\defeq T\setminus N^{\textout}[x]$ for any $x\in X$.
There is an algorithm that takes total $\ot(m+|X| n\ltil)$ time such
that, whp, for every $x\in X$, either 
\begin{itemize}
\item outputs a \emph{kernel} graph $G_{x,T}$ containing $x$ and $t_{x}$
as vertices where $|E(G_{x,T})|= \ot( n\ltil)$ together with a vertex set $Z_{x,T}$ of size $O(k)$
such that a set $Y$ is 
a $(x,t_{x})$-min-separator in $G_{x,T}$ iff $Y\cup Z_{x,T}$ is a $(x,T_{x})$-min-separator in $G$, or
\item certifies that $ T_{x} = \emptyset$ or that there is no
  $(L,S,R)$ satisfying \Cref{eq:balanced cut cond} where $\emptyset \neq T_{x}\subseteq R$,
$\ltil\in[|L|/2,|L|]$ and $x\in L$.
\end{itemize}
\end{lem}

\paragraph{Proof of \Cref{lem:main unbalanced directed}.}
For each $i=1,\dots,\lg(a)$, let $\ltil^{(i)}=2^{i}$.
Let $T^{(i,1)},\dots,T^{(i,O(\log n))}$ be independently obtained
by sampling each vertex with probability $1/(8\ltil^{(i)})$ and let
$X^{(i)}$ be a set of $O(n\log n/\ltil^{(i)})$ random vertices. We
invoke \Cref{lem:kernel directed} with parameters $(\ltil^{(i)},X^{(i)},T^{(i,j)})$
for each $j=1,\dots,O(\log n)$. For each returned graph $G_{x,T^{(i,j)}}$
for some $x\in X_{i}$, we find $(x,t_{x})$-min-separator in $G_{x,T^{(i,j)}}$
by calling the maxflow subroutine and obtain a $(x,T_{x}^{(i,j)})$-min-separator
in $G$ by combining it with $Z_{x,T^{(i,j)}}$. Among all obtained $(x,T_{x}^{(i,j)})$-min-separators
(over all $i,j,x$), we return the one with minimum size as the answer
of \Cref{lem:main unbalanced directed}. If there is no graph $G_{x,T^{(i,j)}}$ returned
from \Cref{lem:kernel directed} at all, then we return an arbitrary vertex
cut. 

Now, we bound the running time. As we call \Cref{lem:kernel directed} $O(\log^{2}n)$ times, this takes  $\ot(\sum_i (m+|X_i|n\ltil^{(i)})) = \ot(m + \sum_i \frac{n}{\ltil^{(i)}} n\ltil^{(i)}) = \ot(n^2)$ time outside max-flow calls. The total time due to max-flow computation is 
\begin{align*} & \sum_{i,j}\sum_{x\in X^{(i)}}t_{\textflow}(|E(G_{x,T^{(i,j)}})|,|E(G_{x,T^{(i,j)}})|)\\
	& =\sum_{i}\ot((n/\ltil^{(i)})\cdot t_{\textflow}(n\ltil^{(i)},n\ltil^{(i)})\cdot)\\
	& =\ot((n/a)\cdot t_{\textflow}(na,na))
\end{align*}
where the last equality is because we assume that $x\cdot t_{\textflow}(m,m)\le t_{\textflow}(xm,xm)$
for any $m,x\ge1$.

To prove correctness, suppose that $G$ has a vertex cut $(L,S,R)$ satisfying \Cref{eq:balanced cut cond}. Consider $i$ such that $\ltil^{(i)}\in[|L|/2,|L|]$. Then, there
exists $x\in X^{(i)}$ where $x\in L$ whp. Also, by \Cref{prop:avoid SL directed},
there is $j$ where $\emptyset \neq T^{(i,j)}_x\subseteq R$ whp. Therefore, a $(x,T_{x}^{(i,j)})$-min-separator
must have size less than $k$ and we must obtain it by \Cref{lem:kernel directed}.  We remark that if the guarantee from \Cref{lem:kernel directed} fails (as it happens with low probability), then the $G_{x,T}$ can be an arbitrary graph, and its cut may not correspond to the actual cut in $G$. To handle this situation,  we add one more checking step to the above algorithm: we verify all of the vertex cuts obtained from the above algorithm. If there exists one that does not correspond to a vertex-cut in $G$, we can return an arbitrary vertex-cut in $G$. The total extra time for verification is $\ot(m)$. Therefore, the algorithm always return a vertex cut of $G$.
This completes the proof.

\paragraph{Organization.}
We formally
show the existence of $G_{x,T}$ in \Cref{sec:Exist Kernel directed} (using
the help of reduction rules shown in \Cref{sec:reduc rule directed}). Next, we give efficient data structures for efficiently building each $G_{x,T}$ in \Cref{sec:Fast Kernel directed} and then use them to finally prove \Cref{lem:kernel directed} in \Cref{sec:kernel proof directed}.

\subsection{Reduction Rules for $(s,t)$-vertex Mincut}
\label{sec:reduc rule directed}

In this section, we describe a simple and generic ``reduction rules''
for reducing the instance size of the $(s,t)$-vertex mincut problem.
We will apply these rules in \Cref{sec:Exist Kernel directed}. Let $H=(V,E)$
be an arbitrary simple directed graph with source $s$ and sink $t$ where $(s,t)\notin E$.

The first rule helps us identify vertices that must be in every mincut and hence we can remove them. %
\begin{prop}
[Identify rule]\label{prop:rule identify directed}Let $H'=H\setminus (N^{\textout}(s)\cap N^{\textin}(t))$.
Then, $S'$ is an $(s,t)$-min-separator in $H'$ if and only if $S=S'\cup(N^{\textout}(s)\cap N^{\textin}(t))$
is an $(s,t)$-min-separator in $H'$.
\end{prop}

\begin{proof}
Let $v\in N^{\textout}(s)\cap N^{\textin}(t)$. Observe that $v$ is contained in \emph{every}
$(s,t)$-separator in $H$. So $S'$ is an $(s,t)$-min-separator
in $H\setminus\{v\}$ iff $S'\cup\{v\}$ is an $(s,t)$-min-separator
in $H$. The claim follows by applying the same argument on another
vertex $v'\in N^{\textout}(s)\cap N^{\textin}(t)\setminus\{v\}$ in $H\setminus\{v\}$
and repeating for all remaining vertices in $N^{\textout}(s)\cap N^{\textin}(t)$.
\end{proof}
The second rule helps us ``filter'' useless edges and vertices w.r.t.
$(s,t)$-vertex connectivity.
\begin{prop}
[Filter rule]\label{prop:rule filter directed}
There exists a maximum set of  $(s,t)$-vertex-disjoint paths $P_1,\dots,P_z$ in $H$ such that no path contains edges/vertices
that satisfies any of the following properties.
\begin{enumerate} [nolistsep,noitemsep]
\item an edge in $\bigcup_{v \in N^{\textout}[s]} \delta^{\textin}(v) \setminus \delta^{\textout}(s)$ or $\bigcup_{v \in N^{\textin}[t]} \delta^{\textout}(v) \setminus \delta^{\textin}(t)$.%
\item a vertex $v$ where $N^{\textout}(v) \ni t$ and
  $N^{\textin}(v) \subseteq N^{\textin}[t]$. %
\item a vertex $v$ where $s$ cannot reach $v$ in $H\setminus N^{\textin}[t]$.
\end{enumerate}
Therefore, by maxflow-mincut theorem, the size of $(s,t)$-vertex mincut in $H$ stays the same
even after we remove these edges and vertices from $H$.
\end{prop}

\begin{proof}
(1): Suppose there is $P_i = (s,\dots,u_{1},u_{2},\dots,t)$ where $(u_{1},u_{2})\in \bigcup_{v \in N^{\textout}[s]} \delta^{\textin}(v) \setminus \delta^{\textout}(s)$. If $u_2 = s$, then we replace $P_i$ with $P'_i = (u_2,\ldots,t)$. Otherwise,  $u_2 \in N^{\textout}(s)$. Thus, we can replace $P_i$ with $P'_i=(s,u_{2},\dots,t)$. Either way, we replace $P_i$ with the new path $P'_i$ that does not use edge $(u_1,u_2)$ and is still disjoint
from other paths $P_j$.   The argument is symmetric for  the set $\bigcup_{v \in N^{\textin}[t]} \delta^{\textout}(v) \setminus \delta^{\textin}(t)$.

(2): Let $v$ be such that   $N^{\textout}(v) \ni t$ and
$N^{\textin}(v) \subseteq N^{\textin}[t]$. We first apply rule
(1). This means that $v$ is unreachable from $s$. %

(3): Suppose $v\in P_i$. There must exist $t'\in N^{\textin}(t)$ where
$P_i=(s,\dots,t',\dots,v,\dots,t)$  because
$s$ could not reach $v$ if $N^{\textin}[t]$ was removed. Then, we can replace
$P_i$ with $P'_i=(s,\dots,t',t)$ which does not contain $v$  and  is still disjoint from other
paths $P_j$. 

\end{proof}

\subsection{Structure of Kernel $G_{x,T}$} %
\label{sec:Exist Kernel directed}

Let $G$ and $k$ be the input of \Cref{lem:main unbalanced directed}.  Throughout this section, we fix a vertex $x$ and a vertex set $T\neq\emptyset$.
The goal of this section is to show the existence of the graph $G_{x,T}$
as needed in \Cref{lem:kernel directed} and state its structural properties
which will be used later in \Cref{sec:Fast Kernel directed,sec:kernel proof directed}. 

Recall that $T_{x}\defeq T\setminus N^{\textout}[x]$. Recall the graph
$G'_{x,T}$ is obtained from $G$ by contracting $T_{x}$ into a
\emph{sink} $t_{x}$. We call $x$ a \emph{source}. Clearly, every
$(x,t_{x})$-vertex cut in $G'_{x,T}$ is a $(x,T_{x})$-vertex cut in $G$.

Let $G_{x,T}$ be obtained from $G'_{x,T}$ by first applying Identify
rule from \Cref{prop:rule identify directed}. Let
$Z_{x,T}=N_{G'_{x,T}}^{\textout}(x)\cap N_{G'_{x,T}}^{\textin}(t_{x})$ be the set removed
from $G'_{x,T}$ by Identify rule. We also write
$Z=Z_{x,T}$ for convenience. After removing $Z$, we apply
Filter rule from \Cref{prop:rule filter directed}. We call the resulting graph the \emph{kernel} graph $G_{x,T}$.
The reduction rules from Propositions \ref{prop:rule identify
  directed} and \ref{prop:rule filter directed} immediately
imply the following.

\begin{lem}
\label{lem:preserve cut directed}A set $Y$ is a $(x,t_{x})$-min-separator
in $G_{x,T}$ iff $Y\cup Z_{x,T}$ is a $(x,T_{x})$-min-separator
in $G$.
\end{lem}

Let us partition vertices of $G_{x,T}$ as follows. Let $N_x=N_{G_{x,T}}^{\textout}(x)$
be the out-neighbor of source $x$. Let $N_t=N_{G_{x,T}}^{\textin}(t_{x})$ be the
in-neighbor of sink $t_{x}$. Note that $N_x$ and $N_t$ are disjoint by Identify rule.
Let $F=V(G_{x,T})\setminus(N_x \cup N_t \cup\{x,t_{x}\})$
be the rest of vertices, which is ``far'' from both $x$ and $t_x$.
By Filter rule(1), every vertex in $N_x$ has only one incoming edge from $x$, and every vertex in $N_t$ has only one outgoing edge to $t_x$. Also, $\delta^{\textin}_{G_{x,T}}(x) \cup \delta^{\textout}_{G_{x,T}}(t_x) = \emptyset.$  Therefore, edges of $G_{x,T}$ can be partitioned to 
\begin{align}
E(G_{x,T}) =  E_{G_{x,T}}(x,N_x)\cup E_{G_{x,T}}(N_x \cup F , F\cup N_t) 
   \cup E_{G_{x,T}}(N_t, t_{x}               ). %
 \label{eq:edge of GxT directed}
\end{align}
Below, we further characterize each part in $G_{x,T}$ in term of sets
in $G = (V,E)$. %

\begin{lem}
\label{lem:structure GxT directed}We have the following:%
\begin{enumerate}
\item $Z=N^{\textout}_G(x)\cap N^{\textin}_G(T_{x})$ and
  $N_x=N^{\textout}_G(x)\setminus N ^{\textin}_G(T_{x})$. So, $Z$ and
$N_x$ partition $N^{\textout}_G(x)$.
\item $F=\{v\in V\setminus(N_G^{\textout}[x]\cup N_G^{\textin}[T_{x}])\mid v$ is reachable from
$N_x$ in $G\setminus N^{\textin}_G[T_{x}]\}$.
\item $N_t=\{v\in N^{\textin}_G(T_{x})\setminus N^{\textout}_G[x]\mid $ there is an edge from $F\cup N_x$ to $v$ in $G \}$.
\end{enumerate}
\end{lem}

\begin{proof}
(1): Observe that $N'_t \defeq  N_{G'_{x,T}}^{\textin}(t_{x}) =
N^{\textin}_G(T_{x})$ and $N'_x \defeq N^{\textout}_{G'_{x,T}}(x) = N^{\textout}(x)$ because $G'_{x,T}$ is simply $G$ after contracting $T_{x}$.
So $Z=N^{\textout}_G(x)\cap N^{\textin}_G(T_{x})$. After removing $Z$ from $G'_{x,T}$ via Identify rule, the
remaining out-neighbor set of $x$ is $N_G^{\textout}(x)\setminus
N_G^{\textin}(T_{x})$.  Since Filter rule never further removes any out-neighbor
of the source $x$, we have $N_x=N^{\textout}_G(x)\setminus N^{\textin}_G(T_{x})$. 

(2): Let $F'=V\setminus(N_G^{\textout}[x]\cup
N_G^{\textin}[T_{x}])$. Note that $F'$ is precisely the set of
vertices in $G'_{x,T}$ that is not an out-neighbor of source $x$ nor an
in-neighbor of sink $t_{x}$. As $F$ is an analogous set for $G_{x,T}$ and $G_{x,T}$ is a subgraph of $G'_{x,T}$, we have $F\subseteq F'$. Observe that only Filter rule(3) may remove vertices from $F'$. (Identify rule and Filter rule(1,2) do not affect $F'$). Now, Filter rule(3)
precisely removes vertices in $F'$ that are not reachable from source
$x$ in $G'_{x,T}\setminus N^{\textin}_{G'_{x,T}}[t_{x}]$. Equivalently, it
removes those that are not reachable from $N_x$ in $G\setminus N^{\textin}[T_{x}]$.
Hence, the remaining part of $F'$ in $G_{x,T}$ is exactly
$F$.    

(3): Let $N''_t=N^{\textin}_G(T_{x})\setminus N^{\textout}_G[x]$. $N''_t$ precisely contains in-neighbors
of sink $t_{x}$ in $G'_{x,T}$ outside $N^{\textout}_G[x]$. As $N_t$ is the neighbor set of $t_{x}$ in $G_{x,T}$ and $Z=N(x)\cap N(T_{x})$ is removed
from $G_{x,T}$, we have that $N_t \subseteq N''_t$. Now, only Filter rule(2)
may remove vertices from $N''_t$, and it precisely removes those that
do not have incoming edge from  $F \cup N_x$. Therefore, the remaining part of $N''_t$
in $G_{x,T}$ is exactly $N_t$.
\end{proof}
Next, we show we bound the size of $E(G_{x,T})$. 
\begin{lem}
\label{lem:small GxT directed} Suppose \Cref{eq:low out deg directed} holds. Then, $|E(G_{x,T})|= O(n\ltil \log n).$%
\end{lem}

\begin{proof}
      Since $G_{x,T}$ has no parallel edges, \Cref{eq:edge of GxT directed} immediately implies  $|E(G_{x,T})| \leq |E_{G_{x,T}}(N_x \cup F, F \cup N_t)| + O(|N_x|+|N_t|) =  |E_{G_{x,T}}(N_x \cup F, F \cup N_t)| + O(n)$. Next, we claim that   $|E_{G_{x,T}}(N_x \cup F, F \cup N_t)| =  O(n\ltil \log n)$. Indeed, for any $v\in V(G_{x,T})\setminus N_t$,
we have $|E_{G_{x,T}}(v,F\cup N_t)|\le40\ltil\ln n$ by \Cref{eq:low out deg directed}.
So $|E_{G_{x,T}}(N_x \cup F, F \cup N_t)| \le(|N_x|+|F|)\cdot40\ltil\ln n = O(n \ltil \log n).$ %
 %
 %
 %
 %
 %
 %
 %
 %
 %
 %
 %
%
%
%
%
%
%
%
 
%
%
%

\end{proof}

\subsection{Data Structures}

\label{sec:Fast Kernel directed}

In this section, we show fast data structures needed for proving \Cref{lem:kernel directed}.
Throughout this section, let $(G,k,\ltil,T)$ denote the input given
to \Cref{lem:kernel directed}. We treat them as global variables in this
section. Moreover, as the guarantee from \Cref{eq:low out deg directed} holds
whp, \textbf{we will assume that \Cref{eq:low out deg directed} holds in this section.}

There are two steps. First, we build an oracle that, given any vertices
$x$ and $v$, lists all neighbors of $v$ outside $N^{\textout}[x]$ if the
set is small. Second, given an arbitrary vertex $x$, we use this
oracle to perform a BFS-like process that allows us to gradually build
$G_{x,T}$ without having an explicit representation of $G_{x,T}$
in the beginning. We show how to solves these tasks respectively
in the subsections below. 

\subsubsection{An Oracle for Listing Out-Neighbors Outside $N^{\textout}[x]$}

In this section, we show the following data structure.
\begin{lem}
[Out-Neighbor Oracle]\label{lem:list neighbor directed}There is an algorithm
that preprocesses $(G=(V,E),k,\ltil)$ in $\ot(m)$ time and supports
queries $\outneigh(x,v)$ for any vertex $x$ where $|N^{\textout}[x]|\le k+2\ltil$
and $v\in V\setminus\{x\}$.

$\outneigh(x,v)$ either returns the neighbor set of $v$ outside
$N^{\textout}[x]$, i.e.~$N^{\textout}(v)\setminus N^{\textout}[x]$, in $\ot(\ltil)$ time or report
``too big'' in $\ot(1)$ time. If $|N^{\textout}(v)\setminus N^{\textout}[x]|\le40\ltil\ln n$, then $N^{\textout}(v)\setminus N^{\textout}[x]$
is returned. If $|N^{\textout}(v)\setminus N^{\textout}[x]|>100\ltil\ln n$, then ``too
big'' is reported. Whp, every query is answered correctly.
\end{lem}

For any vertex set $V'\subseteq V$, let the indicator vector $\one_{V'}\in\{0,1\}^{V}$
of $V'$ be the vector where $\one_{V'}(u)=1$ iff $u\in V'$. In this
section, we always use sparse representation of vectors, i.e.~a list
of (index,value) of non-zero entries of the vector. 

The algorithm preprocesses as follows. Set $s\gets100\ltil\ln n$.
For every vertex $v\in V$, we compute the sketches $\sketch_{s}(\one_{N^{\textout}(v)})$,
$\sketch_{s}(\one_{N^{\textout}[v]})$, $\sketch_{\ltwo}(\one_{N^{\textout}(v)})$, and
$\sketch_{\ltwo}(\one_{N^{\textout}[v]})$ using Theorems  \ref{thm:estimate} and \ref{thm:recovery}.
Observe the following.
\begin{prop}
The preprocessing time is $\ot(m)$.
\end{prop}

\begin{proof}
Theorems \ref{thm:estimate} and \ref{thm:recovery} preprocess in  $\ot(n)$ time. The total time to compute the sketches is $\sum_{v\in V}\ot(\deg^{\textout}(v))=\ot(m)$.
\end{proof}
Now, given a vertex $x$ where $|N^{\textout}[x]|\le k+2\ltil$ and $v\in V\setminus\{x\}$,
observe that the non-zero entries of $\one_{N(v)}-\one_{N[x]}\in\{-1,0,1\}^{V}$
corresponds to the symmetric difference $(N^{\textout}(v)\setminus N^{\textout}[x])\cup(N^{\textout}[x]\setminus N^{\textout}(v))$.
We will bound the size of $N^{\textout}[x]\setminus N^{\textout}(v)$ as follows:
\begin{align*}
  |N^{\textout}[x]\setminus N^{\textout}(v)| & \le|N^{\textout}[x]|-|N^{\textout}[x]\cap N^{\textout}(v)|  \\
  & \le k+2\ltil-(k-|N^{\textout}(v)\setminus N^{\textout}[x]|) \\
  &=|N^{\textout}(v)\setminus N^{\textout}[x]|+2\ltil
\end{align*}
where the second inequality is because $k\le|N^{\textout}(v)|=|N^{\textout}(v)\cap N^{\textout}[x]|+|N^{\textout}(v)\setminus N^{\textout}[x]|$.
Therefore, we have
\[
|N^{\textout}(v)\setminus N^{\textout}[x]|\le\|\one_{N^{\textout}(v)}-\one_{N^{\textout}[x]}\|_{0}\le2|N^{\textout}(v)\setminus N^{\textout}[x]|+2\ltil.
\]
Since  $\one_{N(v)}-\one_{N[x]}\in\{-1,0,1\}^{V}$, we have 
\begin{align} \label{eq:nonzero di}
|N^{\textout}(v)\setminus N^{\textout}[x]|\le\|\one_{N^{\textout}(v)}-\one_{N^{\textout}[x]}\|_{2}\le2|N^{\textout}(v)\setminus N^{\textout}[x]|+2\ltil.
\end{align}

Now, we describe how to answer the query. First, we compute $\sketch_{\ltwo}(\one_{N^{\textout}(v)})-\sketch_{\ltwo}(\one_{N^{\textout}[x]})=\sketch_{\ltwo}(\one_{N^{\textout}(v)}-\one_{N^{\textout}[x]})$
in $O(\log n)$ time. If $\|\sketch_{\ltwo}(\one_{N^{\textout}(v)}-\one_{N^{\textout}[x]})\|_{2}>s$,
then we report ``too big''. Otherwise, we have $s\ge\|\sketch_{\ltwo}(\one_{N^{\textout}(v)}-\one_{N^{\textout}[x]})\|_{2}\ge\|\one_{N^{\textout}(v)}-\one_{N^{\textout}[x]}\|_{2}=\|\one_{N^{\textout}(v)}-\one_{N^{\textout}[x]}\|_{0}$
by \Cref{thm:estimate} and because $\one_{N^{\textout}(v)}-\one_{N^{\textout}[x]}\in\{-1,0,1\}^{V}$.
So, we can compute $\sketch_{s}(\one_{N^{\textout}(v)})-\sketch_{s}(\one_{N^{\textout}[x]})=\sketch_{s}(\one_{N^{\textout}(v)}-\one_{N^{\textout}[x]})$
and obtain the set $N^{\textout}(v)\setminus N^{\textout}[x]$ inside $(N^{\textout}(v)\setminus N^{\textout}[x])\cup(N^{\textout}[x]\setminus N^{\textout}(v))$
 using \Cref{thm:recovery} in $\ot(s)=\ot(\ltil)$ time.

To see the correctness, if $|N^{\textout}(v)\setminus N^{\textout}[x]|\le40\ltil\ln n$,
then 
\[
\|\sketch_{\ltwo}(\one_{N^{\textout}(v)}-\one_{N^{\textout}[x]})\|_{2}\le1.1\|\one_{N^{\textout}(v)}-\one_{N^{\textout}[x]}\|_{2} \overset{(\ref{eq:nonzero di})}{\le}1.1\cdot(2|N^{\textout}(v)\setminus N^{\textout}[x]|+2\ltil)\le100\ltil\ln n=s.
\]
So the set $N^{\textout}(v)\setminus N^{\textout}[x]$ must be returned. If $|N^{\textout}(v)\setminus N^{\textout}[x]|>100\ltil\ln n$,
then $$ \|\sketch_{\ltwo}(\one_{N^{\textout}(v)}-\one_{N^{\textout}[x]})\|_{2} \geq \|\one_{N^{\textout}(v)}-\one_{N^{\textout}[x]}\|_{2} \overset{(\ref{eq:nonzero di})}{\ge}|N^{\textout}(v)\setminus N^{\textout}[x]|>100\ltil\ln n, $$
and so ``too big'' is reported in $\ot(1)$ time. Each query
is correct whp because of the whp guarantees from \Cref{thm:estimate,thm:recovery}. This completes the proof of \Cref{lem:list neighbor directed}.

\subsubsection{Building $G_{x,T}$ by Sketchy Search}

In this section, we show how to use the oracle from \Cref{lem:list neighbor directed}
to return the kernel graph $G_{x,T}$.
\begin{lem}
[Sketchy Search]\label{lem:BFS directed}There is an algorithm that preprocesses
$(G,k,\ltil,T)$ in $\ot(m)$ time and guarantees the following whp.

Given a query vertex $x\in X$, by calling the oracle from \Cref{lem:list neighbor directed}, return either $\bot$ or the kernel
graph $G_{x,T}$ with $O(n\ltil\log n)$ edges together with the set
$Z_{x,T}$ (defined in the beginning of \Cref{sec:Exist Kernel
  directed}) in $\ot(n\ltil)$ time. %
If there is  $(L,S,R)$ satisfying \Cref{eq:balanced cut cond} where $\emptyset \neq T_{x}\subseteq R$,
  $\ltil\in[|L|/2,|L|]$ and $x\in L$, then the algorithm must return $G_{x,T}$ and $Z_{x,T}$.
\end{lem}

The remaining part of this section is for proving \Cref{lem:BFS directed}. 
In the preprocessing step, for each $v \in V$, we just compute
$V_{\bad}=\{v\mid T\subseteq N^{\textout}[v]\}$ by trivially checking
if $T\subseteq N^{\textout}[v]$. This takes total $\sum_{v}\deg(v)= O(m)$ time. Observe that $x\in V_{\bad}$ iff $T_{x}=\emptyset$. Also, if there is  $(L,S,R)$ satisfying \Cref{eq:balanced cut cond} where $\emptyset \neq T_{x}\subseteq R$,
  $\ltil\in[|L|/2,|L|]$ and $x\in L$, then  we must have $|N[x]|\le k+|L|\le k+2\ltil$. So, given a query vertex $x$, if $x\in V_{\bad}$ or $|N^{\textout}[x]|>k+2\ltil$, we can just return $\bot$. From now, we assume that $T_{x}\neq\emptyset$
and $|N^{\textout}[x]|\le k+2\ltil$.

Before showing how to construct $G_{x,T}$, we recall the definitions
of $G_{x,T}$ and $Z_{x,T}$ from \Cref{sec:Exist Kernel directed}. \Cref{eq:edge of GxT directed}
says that edges of $G_{x,T}$ can partitioned as
\begin{align*}
E(G_{x,T}) =  E_{G_{x,T}}(x,N_x)\cup E_{G_{x,T}}(N_x \cup F , F\cup N_t) 
   \cup E_{G_{x,T}}(N_t, t_{x}               ). 
\end{align*}
where $N_x=N^{\textout}_{G_{x,T}}(x)$, $N_t=N^{\textin}_{G_{x,T}}(t_{x})$, and $F=V(G_{x,T})\setminus(N_x\cup N_t\cup\{x,t_{x}\})$.
We write $Z\defeq Z_{x,T}=N^{\textout}_{G'_{x,T}}(x)\cap N^{\textin}_{G'_{x,T}}(t_{x})$
where $G'_{x,T}$ is obtained from $G$ by contracting $T_{x}$ into a single vertex $t_{x}$. 

Our strategy is to exploit $\outneigh$ queries from \Cref{lem:list neighbor directed}
to perform a BFS-like process in $G_{x,T}$ that allows us to gradually
identify $G_{x,T}$ and $Z_{x,T}$ without having an explicit representation
of $G_{x,T}$ in the beginning. The algorithm initializes 
$\Ztil,\Ntil_x,\Etil_{N_x},\Ntil_t,\Ftil,\Etil_{F} =\emptyset$.
At the end of the algorithm, these sets will become
$Z,N_x,E_{G_{x,T}}(N_x,F\cup N_t),N_t,F,E_{G_{x,T}}(F,F\cup N_t)$ respectively. 

Observe that once we know all these sets we can immediately deduce
$E_{G_{x,T}}(x,N_x)$, $E_{G_{x,T}}(N_t,t_{x})$.
Therefore, we obtain all parts in $E(G_{x,T})$. So we can return $G_{x,T}$ and $Z_{x,T}$ as desired. 

The algorithm has two main loops. After the first loop, $\Ztil,\Ntil_x$, and $\Etil_{N_x}$
become $Z,N_x$, and $E_{G_{x,T}}(N_x,F\cup N_t)$ respectively. After the second
loop, $\Ntil_t,\Ftil$, and $\Etil_{F}$ become $N_t,F$, and $E_{G_{x,T}}(F,F\cup N_t)$
respectively. %
Below, we describe this BFS-like process in details.   %

\begin{algorithm}
\begin{enumerate}
\item \label{enu:forloop directed}For each $v\in N^{\textout}(x)$, 
\begin{enumerate}
\item Set $\textsc{visit}(v)=\textsc{true}$.
\item \label{enu:Z toobig directed}If $\outneigh(x,v)$ returns ``too big'', then
add $v$ to $\Ztil$.
\item Else, $\outneigh(x,v)$ returns the set $N^{\textout}(v)\setminus N^{\textout}[x]$. 
\begin{enumerate}
\item \label{enu:Z incident directed}If $N^{\textout}(v)\setminus N^{\textout}[x]$ intersects $T_{x}$,
then add $v$ to $\Ztil$. 
\item \label{enu:N directed}Else, (1) add $v$ to $\Ntil_x$ and edges from $v$
to $N^{\textout}(v)\setminus N^{\textout}[x]$ to $\Etil_{N_x}$, and
(2) add each element in the set $\{w\in N^{\textout}(v)\setminus N^{\textout}[x]\mid\textsc{visit}(w)\neq\textsc{true}\}$
to $\queue$.
\end{enumerate}
\end{enumerate}
\item \label{enu:whileloop directed}While $\exists v\in\queue$,
\begin{enumerate}
\item Remove $v$ from $\queue$. Set $\textsc{visit}(v)=\textsc{true}$. 
\item \label{enu:D too big directed}If $\outneigh(x,v)$ returns ``too big'',
then add $v$ to $\Ntil_t$.
\item Else, $\outneigh(x,v)$ returns the set $N^{\textout}(v)\setminus N^{\textout}[x]$. 
\begin{enumerate}
\item \label{enu:D incident directed}If $N^{\textout}(v)\setminus N^{\textout}[x]$ intersects $T_{x}$,
then add $v$ to $\Ntil_t$.
\item \label{enu:P directed}Else, (1) add $v$ to $\Ftil$ and edges from $v$
to $N^{\textout}(v)\setminus N^{\textout}[x]$ to $\Etil_{F}$, and (2)
add each element in the set $\{w\in N^{\textout}(v)\setminus N^{\textout}[x]\mid\textsc{visit}(w)\neq\textsc{true}\}$
to $\queue$.
\end{enumerate}
\end{enumerate}
\end{enumerate}
\caption{\label{alg:build GxT directed}An algorithm for building $G_{x,T}$}

\end{algorithm}

Before prove the correctness of \Cref{alg:build GxT directed}, we observe the
following simple fact.
\begin{fact}
\label{fact:check incident directed}$N^{\textout}(v)\setminus N^{\textout}[x]$ intersects $T_{x}$
if and only if there is an edge from $v$ to $T_{x}$.
\end{fact}
\begin{proof}
As $T_{x}\cap N^{\textout}[x]=\emptyset$, we have $N^{\textout}(v)\setminus N^{\textout}[x]$
intersects $T_{x}$ iff $N^{\textout}(v)$ intersects $T_{x}$ iff $v\in N^{\textin}(T_{x})$.
\end{proof}
That is, the condition in Steps \ref{enu:Z incident directed} and
\ref{enu:D incident directed}
is equivalent to checking if $v$ is incident to $T_{x}$. Now, we
prove the correctness of the first loop.
\begin{prop}
\label{prop:correct for directed}After the for loop in Step
\ref{enu:forloop directed},
$\Ztil$, $\Ntil_x$, and $\Etil_{N_x}$ become $Z$, $N_x$, and $E_{G_{x,T}}(N_x,F\cup N_t)$,
respectively.
\end{prop}

\begin{proof}
By \Cref{lem:structure GxT directed}(1), $N^{\textout}(x)=Z\dot{\cup}N_x$ where $Z=N^{\textout}(x)\cap N^{\textin}(T_{x})$
and $N_x= N^{\textout}(x)\setminus N^{\textin}(T_{x})$. After the for loop, every $v\in N^{\textout}(x)$
is added to either $\Ztil$ or $\Ntil_x$. If $v$ is added to $\Ztil$
in Step \ref{enu:Z incident directed}, then \Cref{lem:list neighbor directed} implies that $|N^{\textout}(v)\setminus N^{\textout}[x]|>40\ltil\ln n$
and so $v\in N^{\textout}(T_{x})$ by \Cref{eq:low out deg directed}, which means $v\in Z$.
If $v$ is added to $\Ztil$ in Step \ref{enu:Z incident directed}, then we
directly verify that $v\in N^{\textin}(T_{x})$ (see \Cref{fact:check
  incident directed}) and so $v\in Z$ again. Lastly, if $v$ is added to $\Ntil_x$ in Step
\Cref{enu:N directed}, then $v\notin N^{\textin}(T_{x})$ and so $v\in N_x$. This means
that indeed $\Ztil=Z$ and $\Ntil_x=N_x$ after the for loop. Lastly,
every time $v$ is added to $\Ntil_x$, we add $E_{G}(v,V\setminus N^{\textout}[x])=E_{G_{x,T}}(v,F\cup N_x)$
into $\Etil_{N_x}$. So $\Etil_{N_x}$ also collects all edges in $E_{G_{x,T}}(N_x,F\cup N_t)$
after the for loop.
\end{proof}
Next, we prove the correctness of the second loop. The proof is similar to the first one
but more complicated.

\begin{prop}
\label{prop:correct while directed}Suppose $\bot$ is not returned by
\Cref{alg:build GxT directed}.
Then, at end of the while loop in Step \ref{enu:whileloop directed}, $\Ntil_t$,
$\Ftil$ and $\Etil_{F}$ become $N_t$, $F$ and $E_{G_{x,T}}(F,F\cup N_t)$
respectively.
\end{prop}

\begin{proof}
We will prove by induction on time that (1) $\Ntil_t \subseteq N_t$, (2)
$\Ftil\subseteq F$, and (3) if $w\in\textsc{Queue}$ at some point
of time, then $w\in N_t\cup F$.

For the base case, consider the time before the while loop is executed.
We have $\Ntil_t=\emptyset$ and $\Ftil=\emptyset$. If $w\in\queue$,
then $w\in N^{\textout}(v)\setminus N^{\textout}[x]$ for some $v\in N_x$. There are two
cases: if $w\in N^{\textin}(T_{x})$, then $w\in N^{\textin}(T_{x})\setminus N^{\textout}[x]$ and
$w$ has an edge to $v \in N_x$, which means that $w\in N_t$ by
\Cref{lem:structure GxT directed}(3).
Otherwise, if $w\notin N^{\textin}(T_{x})$, then $w\notin N^{\textout}[x]\cup N^{\textin}(T_{x})$
and $(v,w)$ is a path from $N_x$ to $w$ in $G\setminus N^{\textin}[T_{x}]$,
which means that $w\in F$ by \Cref{lem:structure GxT directed}(2). 

For the inductive step, consider that iteration where we visit $v$.
We prove the three statements below one by one. 
\begin{enumerate}
	\item Suppose $v$
	is added to $\Ntil_t$. If $v$ is added at Step \ref{enu:D too
          big directed},
	then \Cref{lem:list neighbor directed} implies that $|N^{\textout}(v)\setminus N^{\textout}[x]|>40\ltil\ln n$ and so $v\in N^{\textin}(T_{x})$
	by \Cref{eq:low out deg directed}. If $v$ is added at Step
        \ref{enu:D incident directed},
	then we directly verify that $v\in N^{\textin}(T_{x})$ (see
        \Cref{fact:check incident directed}).
	In both cases, $v\in N^{\textin}(T_{x})$. As $v\in N_t\cup F$ by induction,
	$v$ must be in $N_t$. So $\Ntil_t\subseteq N_t$ holds.
	\item  Suppose $v$ is added to $\Ftil$, which only happens at
          Step \ref{enu:P directed}.
	We directly verify that $v\notin N^{\textin}(T_{x})$. As $v\in N_t\cup F$ by
	induction, $v$ must be in $F$ and so $\Ftil \subseteq F$ holds.
	\item Suppose
	$w$ is added into $\queue$ at Step \ref{enu:P directed}. There are two cases. If $w\in N^{\textin}(T_{x})$,
	then $w\in N^{\textin}(T_{x})\setminus N^{\textout}[x]$ and
        $w$ has an incoming edge from $v\in F$,
	which means that $w\in N_t$ by \Cref{lem:structure GxT directed}(3). Otherwise,
	if $w\notin N^{\textin}(T_{x})$, then $w\notin N^{\textout}[x]\cup N^{\textin}(T_{x})$. As $v\in F$,
	there exists $p_{v}$ be a path from $N_x$ to $v$ in $G\setminus N^{\textin}[T_{x}]$.
	Now, observe that the path $p_{w}=p_{v}\circ(v,w)$ is a path from
	$N_x$ to $w$ in $G\setminus N^{\textin}[T_{x}]$. Therefore, $w\in F$ by
        \Cref{lem:structure GxT directed}(2).
\end{enumerate}

To show that $\Ntil_t=N_t$ and $\Ftil=F$ at the end, we argue that all
vertices in $N_t \cup F$ must be visited at some point. Observe that our
algorithm simulate a BFS algorithm on $G\setminus Z$ when we start
the search from vertices in $N_x$. Moreover, it never continues the
search once it reaches vertices in $N_t$. By \Cref{lem:structure GxT directed}(2), vertices
in $F$ are reachable from $N_x$ in $G\setminus N^{\textin}[T_{x}]\subseteq G\setminus Z$.
So all vertices from $F$ must be visited. Also, because $N_t \subseteq N^{\textin}(T_{x})\setminus N^{\textout}[x]$
and every vertex in $N_t$ is incident to $F$ or $N_x$, all vertices
from $N_t$ must be visited as well. This completes the proof that $\Ntil_t=N_t$
and $\Ftil=F$ at the end of the while loop. 

Finally, every time $v$ is added to $\Ftil$, we add $E_{G}(v,V\setminus N^{\textout}[x])=E_{G_{x,T}}(v,F\cup N_t)$
into $\Etil_{F}$. So $\Etil_{F}$ collects all edges in $E_{G_{x,T}}(F,F\cup N_t)$
after the while loop.
\end{proof}
Let $v$ be a visited vertex in some iteration of the for loop or
the while loop. We say that $v$'s iteration is \emph{fast} if $\outneigh(x,v)$
returns ``too big'', otherwise we say that $v$'s iteration is
\emph{slow}. 
\begin{prop}
\label{prop:time BFS directed}\Cref{alg:build GxT directed} takes $\ot(\ltil n)$ time.
\end{prop}
\begin{proof}
By the guarantees of \Cref{thm:estimate,thm:recovery}, each
fast iteration takes $\ot(1)$ time, while each slow iteration
takes $\ot(\ltil)$ total time (where checking if $v$ has an edge to
$T_x$ can be done in $O(\log n)$ time using appropriate dictionary
data structure). The number of iterations is clearly at
most $|V| = n$. Therefore, total time is $\ot(\ltil n)$. %
\end{proof}

Now, we conclude with the proof of \Cref{lem:BFS directed}.

\paragraph{Proof of \Cref{lem:BFS directed}.}

Let $(G,k,\ltil,T)$ be given. In the preprocessing step, we only compute
$V_{\bad}$ which takes $O(m)$ time. 
Given a query $x \in V$, if $x\in V_{\bad}$ or $|N^{\textout}[x]|>k+2\ltil$,
we return $\bot$. Then, we execute \Cref{alg:build GxT directed} which takes $\ot(n\ltil)$ time by \Cref{prop:time BFS directed}. The algorithm either returns
$\bot$ and otherwise correctly constructs all parts of $G_{x,T}$ by Propositions \ref{prop:correct for directed} and  \ref{prop:correct while} whp. Using these sets, we can build $G_{x,T}$
via \Cref{eq:edge of GxT directed} and obtain $Z_{x,T}=Z$ in $\ot(n\ltil)$ time.
Note that $|E(G_{x,T})|=O(n\ltil\log n)$ by \Cref{lem:small GxT directed}. 

Finally, if there is  $(L,S,R)$ satisfying \Cref{eq:balanced cut cond} where $\emptyset \neq T_{x}\subseteq R$,
  $\ltil\in[|L|/2,|L|]$ and $x\in L$, then we have $x\notin V_{\bad}$ and $|N^{\textout}[x]|\le k+2\ltil$,
so $\bot$ is not returned before running \Cref{alg:build GxT directed}. %
Therefore,  $G_{x,T}$ and $Z_{x,T}$ must be returned.

\subsection{Proof of \Cref{lem:kernel directed} (Fast  Kernelization)}

\label{sec:kernel proof directed}

Let $(G,k,\ltil,T,X)$ be given as input. We first initialize the oracle from \Cref{lem:list neighbor directed} and the BFS-like process from \Cref{lem:BFS directed}. This takes $\ot(m)$ time. For each
$x\in X$, we query $x$ to the algorithm from \Cref{lem:BFS directed}. \Cref{lem:BFS directed}
guarantees that each query takes $\ot(n\ltil)$ time and returns either
$\bot$ or $(G_{x,T},Z_{x,T})$. Therefore, the total running time
is $\ot(m+|X|n\ltil)$. 

For each query $x\in X$, \Cref{eq:low out deg directed} holds whp by \Cref{prop:low deg whp directed}. So we will
assume it and conclude the following whp. By \Cref{lem:BFS directed}, if $\bot$
is returned, then we can correctly certify that $T_{x}=\emptyset$
or  $(L,S,R)$ satisfying \Cref{eq:balanced cut cond} where $\emptyset \neq T_{x}\subseteq R$,
  $\ltil\in[|L|/2,|L|]$ and $x\in L$. If $(G_{x,T},Z_{x,T})$ is
returned, then we have that $|E(G_{x,T})|=O(n\ltil\log n)$. By \Cref{lem:preserve cut directed}, any set $Y$ is a $(x,t_{x})$-min-separator in $G_{x,T}$ iff
$Y\cup Z_{x,T}$ is a $(x,T_{x})$-min-separator in $G$
as desired.

\end{document}